\documentclass[11pt]{article}
\def\FULL{full}

\usepackage{dsfont}
\usepackage{pstricks,pst-text,pst-node}
\usepackage{amsmath,amssymb}
\usepackage{algorithm,algorithmic}
\usepackage{graphicx}
\usepackage{boxedminipage}
\usepackage{picins}
\usepackage{enumitem}
\usepackage{calc}
\usepackage{ifthen}

\newcommand{\iffull}[1]{\ifthenelse{\equal {\FULL}{full}}{#1}{}}
\newcommand{\ifconf}[1]{\ifthenelse{\equal {\FULL}{full}}{}{#1}}

\newlength{\parcor}

\newlength{\halftw}
\iffull{
  \usepackage{amsthm}
  \usepackage[letterpaper,margin=2.8cm,centering]{geometry}

  \newtheorem{theorem}{Theorem}
  \newtheorem{lemma}{Lemma}
  \newtheorem{claim}{Claim}
  \newtheorem{corollary}{Corollary}

  \newtheorem{example}{Example}

  \setlength{\parcor}{-4ex}
}

\newtheorem{assumption}{Assumption}
\newtheorem{observation}{Observation}

\newcommand{\ignore}[1]{}

\newcommand{\ess}{\ensuremath{\mathcal{S}}}
\newcommand{\su}{\ensuremath{\mathrm{s}}}
\newcommand{\de}{\ensuremath{\mathrm{\pi}}}

\newcommand{\R}{\ensuremath{\mathcal{R}}}

\newcommand{\Sm}{\ensuremath{\mathcal{S}}}
\newcommand{\La}{\ensuremath{\mathcal{L}}}
\newcommand{\C}{\ensuremath{\mathcal{C}}}

\newcommand{\opt}{\ensuremath{\mathtt{opt}}}
\newcommand{\OPT}{\ensuremath{\mathtt{OPT}}}
\newcommand{\NP}{\ensuremath{\text{NP}}}
\newcommand{\IP}{\ensuremath{\text{P}}}

\newcommand{\IZ}{\mathbb{Z}}
\newcommand{\cov}[1]{\ensuremath{\mathtt{Cov}(#1)}}
\newcommand{\bb}{\ensuremath{\bar{b}}}

\newcommand{\bs}{\ensuremath{\bar{s}}}
\newcommand{\bx}{\ensuremath{\bar{x}}}

\newcommand{\val}{\ensuremath{\mathrm{val}}}

\def\ni{\noindent}
\def\xint{x^{\tt int}}

\begin{document}

\title{On Column-restricted and Priority Covering Integer Programs \thanks{
    Supported by NSERC grant no. 288340 and by an Early Research
    Award. Emails: {\tt deeparnab@gmail.com, elyot@uwaterloo.ca, jochen@uwaterloo.ca}}}

\author{
  Deeparnab Chakrabarty
  \and
  Elyot Grant 
  \and
 Jochen K{\"o}nemann 
}

\ifconf{
  \institute{Department of Combinatorics and Optimization \\
    University of Waterloo, Waterloo, ON, Canada N2L 3G1}}
\iffull{
  \date{University of Waterloo, Waterloo, ON, Canada N2L 3G1}
}

\maketitle

\begin{abstract}
  In a column-restricted covering integer program (CCIP), all the
  non-zero entries of any column of the constraint matrix are
  equal. Such programs capture capacitated versions of covering
  problems. In this paper, we study the approximability of CCIPs, in
  particular, their relation to the integrality gaps of the underlying
  0,1-CIP.

  \hspace{5mm} If the underlying 0,1-CIP has an integrality gap
  $O(\gamma)$, and assuming that the integrality gap of the {\em
    priority version} of the 0,1-CIP is $O(\omega)$, we give a factor
  $O(\gamma + \omega)$ approximation algorithm for the CCIP. Priority
  versions of 0,1-CIPs (PCIPs) naturally capture {\em quality of
    service} type constraints in a covering problem.

  \hspace{5mm} We investigate priority versions of the line (PLC) and
  the (rooted) tree cover (PTC) problems.  Apart from being natural
  objects to study, these problems fall in a class of fundamental
  geometric covering problems. We bound the integrality of certain
  classes of this PCIP by a constant.  Algorithmically, we give a
  polytime exact algorithm for PLC, show that the PTC problem is
  APX-hard, and give a factor $2$-approximation algorithm for it.

\end{abstract}


\section{Introduction}
In a {\em 0,1-covering integer program} (0,1-CIP, in short), we are
given a constraint matrix $A \in \{0,1\}^{m \times n}$, demands $b \in
\IZ^m_+$, non-negative costs $c \in \IZ_+^n$, and upper bounds $d \in
\IZ^n_+$, and the goal is to solve the following integer linear
program (which we denote by \cov{A,b,c,d}).

$$ \min \{ c^Tx \,:\, Ax \geq b, 0 \leq x \leq d, x \mbox{ integer}\}. $$

Problems that can be expressed as 0,1-CIPs are essentially equivalent
to set multi-cover problems, where sets correspond to columns and
elements correspond to rows.  This directly implies that 0,1-CIPs are
rather well understood in terms of approximability: the class admits
efficient $O(\log n)$ approximation algorithms and this is best
possible unless $\NP=\IP$.  Nevertheless, in many cases one can get
better approximations by exploiting the structure of matrix $A$.  For
example, it is well known that whenever $A$ is {\em totally
  unimodular} (TU)(e.g., see \cite{Sc03}), the canonical LP relaxation
of a 0,1-CIP is integral; hence, the existence of efficient algorithms
for solving linear programs immediately yields fast exact algorithms
for such 0,1-CIPs as well.

While a number of general techniques have been developed for obtaining
improved approximation algorithms for structured $0,1$-CIPs, not much is known
for structured non-$0,1$ CIP instances. In this paper, we attempt to
mitigate this problem, by studying the class of {\em column-restricted
  covering integer programs} (CCIPs), where all the non-zero entries 
  of any column of the constraint matrix are equal. Such CIPs arise
  naturally out of $0,1$-CIPs, and the main focus of this paper 
  is to understand how the structure of the underlying 0,1-CIP can be used 
  to derive improved approximation algorithms for CCIPs.
  \\

\noindent
 {\bf Column-Restricted Covering IPs (CCIPs):}  
Given a 0,1-covering problem \cov{A,b,c,d} and a supply vector
$s\in \IZ^n_+$, the corresponding CCIP is obtained as follows. Let
$A[s]$ be the matrix obtained by replacing all the $1$'s in the $j$th
column by $s_j$; that is, $A[s]_{ij} = A_{ij}s_j$ for all $1\le i\le
m, 1\le j\le n$. The column-restricted covering problem is given 
by the following integer program.

\begin{equation}\tag{\cov{A[s],b,c,d}}
 \min \{c^Tx \,:\, A[s]x \geq b, 0 \leq x \leq d, x \mbox{ integer}\}.
\end{equation}
\noindent
CCIPs naturally capture {\em capacitated} versions of 0,1-covering 
problems. To illustrate this we use the following 0,1-covering problem
called the tree covering problem.
The input is a tree
 $T=(V,E)$ rooted at a vertex $r\in V$, a set of {\em segments} ${\cal
    S} \subseteq \{(u,v): u \textrm{ is a child of}~ v\}$,
  non-negative costs $c_j$ for all $j \in \ess$, and demands $b_e
  \in \IZ_+$ for all $e \in E$. An edge $e$ is contained in a segment
  $j=(u,v)$ if $e$ lies on the unique $u,v$-path in $T$.  The goal is
  to find a minimum-cost subset $C$ of segments such that each edge $e
  \in E$ is contained in at least $b_e$ segments of $C$.  When $T$
  is just a line, we call the above problem, the {\em line cover} (LC)
  problem.  In this example, the constraint matrix $A$ has a row for
  each edge of the tree and a column for each segment in $\ess$. 
  It is not too hard to show that this matrix is TU and thus these can be solved exactly
  in polynomial time.

In the above tree cover problem, suppose each segment
$j\in \ess$ also has a capacity supply $s_j$ associated with it, and
call an edge $e$ covered by a collection of segments $C$ iff the total
supply of the segments containing $e$ exceeds the demand of $e$. 
The problem of finding the minimum cost subset of segments covering 
every edge is precisely the column-restricted tree cover problem. The column-restricted 
line cover problem encodes the minimum knapsack problem and is thus 
NP-hard. \\
 
\noindent
For general CIPs, the best known approximation algorithm, due to
Kolliopoulos and Young~\cite{KY05}, has a performance guarantee of
$O(1+ \log \alpha)$, where $\alpha$, called the {\em dilation} of the instance,
denotes the maximum number of non-zero entries in any column of
the constraint matrix. Nothing better is known for the special case of CCIPs
unless one aims for {\em bicriteria} results where solutions 
violate the upper bound constraints $x \leq d$ (see Section \ref{sec:rel-work} for more details). 

In this paper, our main aim is to understand how the approximability
of a given CCIP instance is determined by the structure of the
underlying $0,1$-CIP. In particular, if a $0,1$-CIP has a constant
integrality gap, under what circumstances can one get constant factor
approximation for the corresponding CCIP?  We make
some steps toward finding an answer to this question.

%

In our main result, we show that there is a constant factor
approximation algorithm for CCIP if {\em two} induced $0,1$-CIPs have
constant integrality gap.  The first is the underlying original
0,1-CIP.  The second is a {\em priority} version of the 0,1-CIP (PCIP, in
short), whose constraint matrix is derived from that of the 0,1-CIP
as follows. \\

\noindent
{\bf Priority versions of Covering IPs (PCIPs):}
Given a 0,1-covering problem \cov{A,b,c,d}, a priority supply
vector $s \in \IZ_+^n$, and a priority demand vector $\pi \in \IZ^m_+$,
the corresponding PCIP is as follows. Define $A[s,\pi]$ to be the
following 0,1 matrix
\begin{equation}\label{eq:trcA}
  A[s,\pi]_{ij} = \left\{\begin{array}{l@{\quad:\quad}l}
      1 & A_{ij}=1 \mbox{ and } s_j \geq \pi_i \\
      0 & \mbox{otherwise,} \end{array}\right.
\end{equation}
Thus, a column $j$ covers row $i$, only if its priority supply is
higher than the priority demand of row $i$.  The priority covering
problem is now as follows.
\begin{equation}\tag{\cov{A[s,\pi],\mathds{1},c}}
  \min \{c^Tx \,:\, A[s,\pi]x \geq \mathds{1}, 0 \leq x \leq d, x
  \mbox{ integer}\}.
\end{equation}

We believe that priority covering problems are interesting in their
own right, and they arise quite naturally in covering applications
where one wants to model {\em quality of service} (QoS) or priority
restrictions. For instance,  in the tree cover problem
defined above,
suppose each segment $j$ has a {\em quality of
  service} (QoS) or priority supply $s_j$ associated with it and suppose each edge
$e$ has a QoS or priority demand $\pi_e$ associated with it. We say that a segment
$j$ covers $e$ iff $j$ contains $e$ {\em and} the priority supply of $j$
exceeds the priority demand of $e$. The goal is to find a minimum cost
subset of segments that covers every edge. This is the priority tree cover problem.

Besides being a natural covering problem to study, we show that the
priority tree cover problem is a special case of a classical geometric
covering problem: that of finding a minimum cost cover of points by
axis-parallel rectangles in $3$ dimensions.  Finding a constant factor
approximation algorithm for this problem, even when the rectangles
have uniform cost, is a long standing open problem. 

We show that although the tree cover is polynomial time solvable, the priority
tree cover problem is APX-hard. 
We complement this with a factor $2$ approximation for
the problem. Furthermore, we present constant upper bounds for
the integrality gap of this PCIP in a number of special cases, implying
constant upper bounds on the corresponding CCIPs in these special cases.
We refer the reader to Section \ref{sec:tech} for a formal statement of our results,
which we give after summarizing works related to our paper.

\subsection{Related work}\label{sec:rel-work}

There is a rich and long line of work
(\cite{Do82,Ho82,RV93,Sr99,Sr06}) on approximation algorithms for
CIPs, of which we state the most relevant to our work.  Assuming no
upper bounds on the variables, Srinivasan \cite{Sr99} gave a $O(1 +
\log \alpha)$-approximation to the problem (where $\alpha$ is the
dilation as before).  Later on, Kolliopoulos and Young \cite{KY05}
obtained the same approximation factor, respecting the upper
bounds. However, these algorithms didn't give any better results when
special structure of the constraint matrix was known.  On the hardness
side, Trevisan~\cite{Tr01} showed that it is NP-hard to obtain a
$(\log \alpha - O(\log\log \alpha))$-approximation algorithm even for
0,1-CIPs.

The most relevant work to this paper is that of Kolliopoulos
\cite{Ko03}. The author studies CCIPs which satisfy a 
rather strong assumption, called the {\em no bottleneck
  assumption}, that the supply of any column is smaller than the
demand of any row. Kolliopoulos \cite{Ko03} shows that 
 if one is allowed to
violate the upper bounds by a multiplicative constant, then the
integrality gap of the CCIP is within a constant factor of that of the
original 0,1-CIP\footnote{Such a result is implicit in the paper;
  the author only states a $O(\log \alpha)$ integrality gap.}. As the
author notes such a violation is necessary; otherwise the CCIP has
unbounded integrality gap. If one is not allowed to violated upper
bounds, nothing better than the result of \cite{KY05} is known for
the special case of CCIPs. 

 Our work on CCIPs parallels a large body of work on column-restricted
 {\em packing} integer programs (CPIPs). Assuming the {\em
   no-bottleneck assumption}, Kolliopoulos and Stein \cite{KS04} show
 that CPIPs can be approximated asymptotically as well as the
 corresponding 0,1-PIPs. Chekuri et al.~\cite{CMS07} subsequently
 improve the constants in the result from \cite{KS04}. These results
 imply constant factor approximations for the column-restricted tree
 {\em packing} problem under the no-bottleneck assumption. Without the
 no-bottleneck assumption, however, only polylogarithmic approximation
 is known for the problem \cite{CEK09}.

The only work on priority versions of covering problems that we are
aware of is due to Charikar, Naor and Schieber~\cite{CNS04} who
studied the priority Steiner tree and forest problems in the context of QoS
management in a network multicasting application.  Charikar et
al. present a $O(\log n)$-approximation algorithm for the problem, and
Chuzhoy et al.~\cite{CG+08} later show that no efficient $o(\log\log
n)$ approximation algorithm can exist unless $\NP \subseteq
\text{DTIME}(n^{\log\log\log n})$ ($n$ is the number of vertices).

To the best of our knowledge, the column-restricted or priority
versions of the line and tree cover problem have not been studied. The
best known approximation algorithm known for both is the $O(\log n)$
factor implied by the results of \cite{KY05} stated above. However,
upon completion of our work, Nitish Korula \cite{Ko09} pointed out to
us that a $4$-approximation for column-restricted line cover is
implicit in a result of Bar-Noy et al. \cite{BarNoy}.  We remark that
their algorithm is not LP-based, although our general result on CCIPs
is.

\subsection{Technical Contributions and Formal Statement of Results}\label{sec:tech}
\noindent
Given a 0,1-CIP \cov{A,b,c,d}, we obtain its {\em canonical LP
  relaxation} by removing the integrality constraint.  The {\em
  integrality gap} of the CIP is defined as the supremum of the ratio
of optimal IP value to optimal LP value, taken over all non-negative 
integral vectors $b,c$, and $d$. The integrality gap of an IP
captures how much the integrality constraint affects the optimum, and
is an indicator of the {\em strength} of a linear programming
formulation.  \\

\noindent
{\bf CCIPs: }
Suppose the CCIP is \cov{A[s],b,c,d}.
We make the following two assumptions about the integrality gaps 
of the 0,1 covering programs, both the original 0,1-CIP and the 
priority version of the 0,1-CIP.

\begin{assumption}\label{as:1}
  The integrality gap of the original 0,1-CIP is $\gamma \ge
  1$. Specifically, for any non-negative integral vectors
  $b,c$, and $d$, if the canonical LP
  relaxation to the CIP has a fractional solution $x$, then one can
  find in polynomial time an integral feasible solution to the CIP of
  cost at most $\gamma \cdot c^T x$. We stress here that the entries
  of $b,c,d$ could be $0$ as well as $\infty$.
\end{assumption}

\begin{assumption}\label{as:2}
  The integrality gap of the PCIP is $\omega \ge 1$.  Specifically,
  for any non-negative integral vectors $s,\pi,c$, if the canonical LP
  relaxation to the PCIP has a fractional solution $x$, then one can
  find in polynomial time, an integral feasible solution to the PCIP
  of cost at most $\omega \cdot c^T x$.
\end{assumption}

We give an LP-based approximation algorithm for solving CCIPs.  Since
the canonical LP relaxation of a CCIP can have unbounded integrality
gap, we strengthen it by adding a set of valid constraints called the
{\em knapsack cover constraints}.  We show that the integrality gap of
this strengthened LP is $O(\gamma + \omega)$, and can be used to give
a polynomial time approximation algorithm.

\begin{theorem}\label{thm:1}
  Under Assumptions \ref{as:1} and \ref{as:2}, there is a
  $(24\gamma+8\omega)$-approximation algorithm for column-restricted
  CIPs. 
\end{theorem}

Knapsack cover constraints to strengthen LP relaxations were
introduced in \cite{B75,HJP75,Wo75}; Carr et al. \cite{CF+00} were the
first to employ them in the design approximation algorithms. The paper
of Kolliopoulos and Young \cite{KY05} also use these to get their
result on general CIPs.

The main technique in the design of algorithms for column-restricted
problems is {\em grouping-and-scaling} developed by Kolliopoulos and
Stein \cite{KS01,KS04} for packing problems, and later used by
Kolliopoulos \cite{Ko03} in the covering context.  In this technique,
the {\em columns} of the matrix are divided into groups of `close-by'
supply values; in a single group, the supply values are then scaled to
be the same; for a single group, the integrality gap of the original
0,1-CIP is invoked to get an integral solution for that group; the
final solution is a `union' of the solutions over all groups.

There are two issues in applying the technique to the new strengthened
LP relaxation of our problem.  Firstly, although the original
constraint matrix is column-restricted, the new constraint matrix with
the knapsack cover constraints is not.  Secondly, unless additional
assumptions are made, the current grouping-and-scaling analysis
doesn't give a handle on the degree of violation of the upper bound
constraints. This is the reason why Kolliopoulos \cite{Ko03} needs the
strong no-bottleneck assumption.

We get around the first difficulty by grouping the {\em rows} as well,
into those that get most of their coverage from columns not affected
by the knapsack constraints, and the remainder. On the first group of
rows, we apply a subtle modification to the vanilla
grouping-and-scaling analysis and obtain a $O(\gamma)$
approximate feasible solution satisfying these rows; we then show that
one can treat the remainder of the rows as a PCIP and get a
$O(\omega)$ approximate feasible solution satisfying them, using
Assumption 2. Combining the two gives the $O(\gamma + \omega)$ factor.
The full details are given in Section 2.

We stress here that apart from the integrality gap
assumptions on the 0,1-CIPs, we do not make any other assumption
(like the no-bottleneck assumption). In fact, we can use the modified
analysis of the grouping-and-scaling technique to get a similar result
as \cite{Ko03} for approximating CCIPs violating the upper-bound
constraints, under a {\em weaker} assumption than the no-bottleneck
assumption. The no-bottleneck assumption states that the
supply of {\em any} column is less than the demand of {\em any} row.
In particular, even though a column has entry $0$ on a certain row,
its supply needs to be less than the demand of that row. We show that
if we weaken the no-bottleneck assumption to assuming that the supply
of a column $j$ is less than the demand of any row $i$ only if
$A[s]_{ij}$ is positive, a similar result can be obtained via our modified analysis.

\begin{theorem}\label{thm:2}
  Under assumption \ref{as:1} and assuming $A_{ij} s_j \le b_i$, for
  all $i,j$, given a fractional solution $x$ to the canonical LP
  relaxation of \cov{A[s],b,c,d}, one can find an integral solution
  $\xint$ whose cost $c\cdot \xint \le 10\gamma (c\cdot x)$ and $\xint
  \le 10d$.
\end{theorem}

\paragraph{Priority Covering Problems.}
In the following, we use PLC and PTC to refer to the priority versions
of the line cover and tree cover problems, respectively. 
Recall that the constraint matrices for line and tree cover problems
are totally unimodular, and the integrality of the corresponding 
0,1-covering problems is therefore $1$ in both case. It is
interesting to note that the 0,1-coefficient matrices for PLC and 
PTC are not totally unimodular in general. 
The following integrality gap bound is obtained via a primal-dual
algorithm.

\begin{theorem}\label{thm:plc-gap}
  The canonical LP for priority line cover has an integrality gap of at least $3/2$ and 
  at most $2$.
\end{theorem}

In the case of tree cover, we obtain constant upper bounds on
the integrality gap for the case $c=\mathds{1}$, that is, for the
minimum cardinality version of the problem.  We believe that the PCIP
for the tree cover problem with general costs also has a constant
integrality gap.  On the negative side, we can show an integrality gap
of at least $\frac{e}{e-1}$.

\begin{theorem}\label{thm:ptc-gap}
  The canonical LP for {\em unweighted} PTC has an integrality gap of
  at most $6$.
\end{theorem}

We obtain the upper bound by taking a given PTC instance and a
fractional solution to its canonical LP, and decomposing it into a
collection of PLC instances with corresponding fractional solutions,
with the following two properties. First, the total cost of the
fractional solutions of the PLC instances is within a constant of the
cost of the fractional solution of the PTC instance. Second, union of
integral solutions to the PLC instances gives an integral solution to
the PTC instance. The upper bound follows from Theorem
\ref{thm:plc-gap}.  Using Theorem \ref{thm:1}, we get the following as
an immediate corollary.

\begin{corollary}
  There are $O(1)$-approximation algorithms for column-restricted line
  cover and the cardinality version of the column-restricted tree
  cover.
\end{corollary}

We also obtain the following combinatorial results.

\begin{theorem}\label{thm:plc-exact}
  There is a polynomial-time exact algorithm for PLC.
\end{theorem}

\begin{theorem}\label{thm:ptc-hard}
  PTC is APX-hard, even when all the costs are unit.
\end{theorem}

\begin{theorem}\label{thm:ptc-apx}
  There is an efficient $2$-approximation algorithm for 
  PTC. 
\end{theorem}

The algorithm for PLC is a non-trivial dynamic programming approach
that makes use of various structural observations about the optimal
solution.  The approximation algorithm for PTC is obtained via a
similar decomposition used to prove Theorem \ref{thm:ptc-gap}.

We end by noting some interesting connections between the priority
tree covering problem and set covering problems in computational
geometry. The {\em rectangle cover} problem in $3$-dimensions is the
following: given a collection of points $P$ in $\mathbb{R}^3$, 
and a collection $C$ of axis-parallel rectangles with
costs, find a minimum cost collection of rectangles that covers every
point. We believe studying the PTC problem
could give new insights into the rectangle cover problem.

\begin{theorem}\label{thm:ptc-geom}
The priority tree covering problem is a special case of the rectangle cover problem in $3$-dimensions.
\end{theorem}

\ifconf{Due to space restrictions, we omit many proofs. A full version of the paper 
is available \cite{CGK-full}.}

\section{General Framework for Column Restricted CIPs}
\label{sec:ccip}

In this section we prove Theorem \ref{thm:1}.
Our goal is to round a solution to a LP relaxation of 
\cov{A[s],b,c,d} into an approximate integral solution.  We strengthen the following
canonical LP relaxation of the CCIP
$$\min \{c^Tx ~:~ A[s]x \ge b, 0\le x\le d, x\ge 0\}$$
by adding valid {\em knapsack cover} constraints. 
In the following we use \C\ for the set of columns and \R\
for the set of rows of $A$.

\subsection{Strengthening the canonical LP Relaxation}

Let $F\subset \C$ be a subset of the columns in the column restricted
CIP \cov{A[s],b,c,d}.  For all rows $i\in \R$, define
$ b^F_i=\max\{0,b_i - \sum_{j \in F} A[s]_{ij}d_j\} $
to be the residual demand of row $i$ w.r.t. $F$. Define matrix $A^F[s]$ by letting
\begin{equation}\label{eq:defAf}
A^F[s]_{ij}=\left\{\begin{array}{l@{\quad :\quad}l}
    \min\{A[s]_{ij},b^F_i\} & j \in \C\setminus F \\
    0 & j \in F, \end{array}\right. \end{equation}
for all $i \in \C$ and for all $j \in \R$. 
The following {\em Knapsack-Cover} (KC) inequality 
$$ \sum_{j\in\C} A^F[s]_{ij}x_j \geq b^F_i $$
is valid for the set of all integer solutions $x$ for
$\cov{A[s],b,c,d}$.  Adding the set of all KC inequalities yields the
following stronger LP formulation CIP. We note that the LP is not column-restricted, in
that, different values appear on the same column of the new constraint matrix.

\begin{align}
\opt_P := \min \quad & \sum_{j \in \C} c_jx_j \tag{P} \label{lp} \\
\mbox{s.t.} \quad  & \sum_{j \in \C} A^F[s]_{ij} x_j \geq b^F_i \quad & \forall F\subseteq \C, \forall i \in \R \label{kc} \\
& 0 \leq x_j \leq d_j & \forall j \in \C \notag
\end{align}
\ni It is not known whether \eqref{lp} can be solved in polynomial
time.  For $\alpha \in (0,1)$, call a vector $x^*$ $\alpha$-relaxed if
its cost is at most $\opt_P$, and if it satisfies \eqref{kc} for $F=\{
j \in \C \,:\, x^*_j \geq \alpha d_j \}$. An $\alpha$-relaxed solution
to \eqref{lp} can be computed efficiently for any $\alpha$. To see
this note that one can check whether a candidate solution satisfies
\eqref{kc} for a set $F$; we are done if it does, and otherwise we
have found an inequality of \eqref{lp} that is violated, and we can
make progress via the ellipsoid method. Details can be found in
\cite{CF+00} and \cite{KY05}.

We fix an $\alpha \in (0,1)$, specifying its precise value later.
Compute an $\alpha$-relaxed
solution,  $x^*$, for \eqref{lp}, and let $F=\{ j \in \C \,:\, x^*_j \geq \alpha d_j \}$. Define $\bar{x}$ as,
$\bar{x}_j = x^*_j$ if $ j\in \C\setminus F$, and $\bar{x}_j = 0$, otherwise.
Since $x^*$ is an $\alpha$-relaxed solution, we get that $\bar{x}$ is
a feasible fractional solution to the {\em residual} CIP,
\cov{A^F[s],b^F,c,\alpha d}. In the next subsection, our goal will be
to obtain an {\em integral} feasible solution to the covering problem
\cov{A^F[s],b^F,c,d} using $\bar{x}$. The next lemma shows how this
implies an approximation to our original CIP.

\def\xint{x^{\tt int}}
\begin{lemma} \label{lem:res}
 	If there exists an integral feasible solution, $\xint$, to \cov{A^F[s],b^F,c,d}  with $c^T\xint \le \beta\cdot c^T\bar{x}$, then there
	exists a $\max\{1/\alpha, \beta\}$-factor approximation to \cov{A[s],b,c,d}.
\end{lemma}
\iffull{
\begin{proof}
Define 
\begin{equation}\label{eq:zdef}
  z_j = \left\{\begin{array}{l@{\quad :\quad}l}
    d_j & j \in F \\
    \xint_j & j \in \C\setminus F, \end{array}\right. 
\end{equation}
Observe that $z\le d$. 
$z$ is a feasible integral solution to  \cov{A[s],b,c,d} since for any $i\in \R$,
\begin{align*}
\sum_{j\in \C} A[s]_{ij}z_j = \sum_{j\in F} A[s]_{ij}d_j + \sum_{j\in\C\setminus F} A[s]_{ij}\xint_j 
                                   \ge (b_i - b^F_i) + \sum_{j\in\C\setminus F} A^F[s]_{ij}\xint_j 
                                  \ge b_i 
\end{align*}
where the first inequality follows from the definition of $b^F_i$ and since $A[s]_{ij} \ge A^F[s]_{ij}$, the second inequality follows since 
$\xint$ is a feasible solution to  \cov{A^F[s],b^F,c,d}. 

\ni
Furthermore, 
$$c^Tz = \sum_{j\in F}c_j d_j + \sum_{j\in \C\setminus F} c_j \xint_j \le \frac{1}{\alpha} \sum_{j\in F}c_j x^*_j + \beta \sum_{j\in \C\setminus F} c_j x^*_j 
\le \max\{\frac{1}{\alpha},\beta\} \opt_P$$ 
where the first inequality follows from the definition of $F$ and the second from the assumption in the theorem statement.
\end{proof}
}
\subsection{Solving the Residual Problem}
In this section we use a feasible fractional solution $\bar{x}$ of \cov{A^F[s],b^F,c,\alpha d}, to obtain an {\em integral} feasible solution $\xint$ to  the covering problem \cov{A^F[s],b^F,c,d}, with $c^T\xint \le \beta c^T\bar{x}$ for $\beta = 24\gamma + 8\omega$. Fix $\alpha = 1/24$. \\

\noindent
{\bf \em Converting to Powers of $2$. } 
For ease of exposition, we first modify the input to the residual problem \cov{A^F[s],b^F,c,d} so that all entries of are powers of $2$. For every $i \in \R$, let $\bb_i$ denote the smallest power of $2$ larger than $b^F_i$.
For every column $j\in\C$, let $\bs_j$ denote the largest power of $2$ smaller than $s_j$.

\begin{lemma} \label{lem:yfeas}
  $y=4\bx$ is feasible for \cov{A^F[\bs],\bb,c,4\alpha d}.
\end{lemma}
\iffull{
\begin{proof}
  Focus on row $i \in \R$. We have
  $$ 
  \sum_{j \in \C} A^F[\bs]_{ij} y_j \geq 2\cdot \sum_{j \in \C} A^F[s]_{ij} \bx_j 
  \geq 2 b^F_i \geq \bb_i,
  $$
  where the first inequality uses the fact that $s_j \leq 2\bs_j$ for
  all $j \in \C$, the second inequality uses the fact that \bx\ is
  feasible for
  \cov{A^F[s],b^F,c,\alpha d}, and the third follows from the definition of $\bb_i$. 
\end{proof}
}
\noindent
{\bf \em Partitioning the rows. } 
We call $\bb_i$ the residual demand of row $i$. For a row $i$, a column $j\in \C$ is {\em $i$-large} if the supply of $j$
is at least the residual demand of row $i$; it is {\em $i$-small} otherwise. Formally,
\begin{eqnarray*}
  \La_i & = & \{ j \in \C\,:\, A_{ij}=1, \bs_j \geq \bb_i \} ~~ \mbox { is the set of $i$-large columns} \\
  \Sm_i & = & \{ j \in \C\,:\, A_{ij}=1, \bs_j < \bb_i \} ~~ \mbox { is the set of $i$-small columns} 
\end{eqnarray*}
Recall the definition from \eqref{eq:defAf}, $A^F[\bs]_{ij} = \min(A[\bs]_{ij}, b^F_i)$.
Therefore, $A^F[\bs]_{ij} = A_{ij}b^F_i$ for all $j\in \La_i$ since $\bs_j \ge \bb_i \ge b^F_i$; 
and $A^F[\bs]_{ij} = A_{ij}\bs_j$ for all $j\in \Sm_i$, since being powers of $2$, 
$\bs_j < \bb_i$ implies, $\bs_j \le \bb_i/2 \le b^F_i$. 

We now partition the rows into large and small depending on which columns
most of their coverage comes from. Formally, call a row $i \in \R$ {\em large} if
$$ \sum_{j \in \Sm_i} A^F[\bs]_{ij}y_j \leq \sum_{j \in \La_i}A^F[\bs]_{ij}y_j, $$
and small otherwise. Note that Lemma \ref{lem:yfeas} together with the
fact that each column in row $i$'s support is either small or large implies,
\begin{align*}
\mbox{For a large row $i$,} ~~ \sum_{j\in \La_i} A^F[\bs]_{ij}y_j \ge \bb_i/2 ,~~
\mbox{For a small row $i$,} ~~ \sum_{j\in \Sm_i} A^F[\bs]_{ij}y_j \ge \bb_i/2 
\end{align*}
Let $\R_L$ and $\R_S$ be the set of large and small rows. 

\def\xints{x^{{\tt int},\Sm}}
\def\xintl{x^{{\tt int},\La}}
In the following, we address small and large rows separately. We
compute a pair of integral solutions $\xints$ and $\xintl$ that are
feasible for the small and large rows, respectively. We then obtain
$\xint$ by letting
\begin{equation}\label{eq:hxdef}
  \xint_j = \max\{\xints_j,\xintl_j\},
\end{equation}
for all $j \in \C$. 

\subsubsection{Small rows.} 
For these rows we use the grouping-and-scaling technique a la
\cite{CMS07,Ko03,KS01,KS04}.  However, as mentioned in the
introduction, we use a modified analysis that bypasses the
no-bottleneck assumptions made by earlier works.

\begin{lemma}\label{lem:small-rows}
We can find an integral solution $\xints$ such that \\
\indent
a) $\xints_j \le d_j$ for all $j$, \\ \indent
b) $\sum_{j\in \C}c_j\xints_j \le 24\gamma \sum_{j\in \C}c_j\bx_j$, and \\ \indent
c) for every  small row $i\in \R_S$, $\sum_{j\in \C} A^F[s]_{ij}\xints_j \ge b^F_i$.
\end{lemma}
\begin{proof} 
\iffull{
The complete proof is slightly technical and hence we start with a sketch.}
\ifconf{(Sketch)} 
Since the rows are small, for any row $i$, we can zero out the entries 
that are larger than $\bb_i$, and still $2y$ will be a feasible solution. 
Note that, now in each row, the entries are $< \bb_i$, and thus are at most
$\bb_i/2$ (everything being powers of $2$). We stress that it could be that $\bb_i$ 
of some row is less than the entry in some other row, that is, we don't have the 
no-bottleneck assumption. However, when a particular row $i$ is fixed, $\bb_i$ is 
at least any entry of the matrix in the $i$th row. 
Our modified analysis of grouping and scaling then makes the proof go through.

We {\em group} the columns into classes that have $s_j$ as the same power of $2$,
and for each row $i$ we let $\bb_i^{(t)}$ be the contribution of the class $t$ columns 
towards the demand of row $i$. The columns of class $t$, the small rows, and the 
demands $\bb_i^{(t)}$ form a CIP where all non-zero entries of the matrix are the same
power of $2$. We scale both the constraint matrix and $\bb^{(t)}_i$  down by that 
power of $2$ to get a 0,1-CIP, and using assumption 1, we get an integral solution
to this 0,1-CIP. Our final integral solution is obtained by concatenating all these integral 
solutions over all classes. 

Till now the algorithm is the standard grouping-and-scaling algorithm.
The difference  lies in our analysis  in proving that this integral solution is feasible for the original 
CCIP. Originally the no-bottleneck assumption was used to prove this. However, we
show since the column values in different classes 
are geometrically decreasing, the weaker assumption of $\bb_i$ being
at least any entry in the $i$th row is enough to make the analysis 
go through. \ifconf{This completes the sketch of the proof.} \iffull{We now get into the full proof.\\

\vspace*{\parcor}
\paragraph{Step 1: Grouping the columns.}

Let $\bs_{min}$ and $\bs_{max}$ be the smallest and largest supply among 
the columns in $\C\setminus F$. Since all $\bs_j$ are powers of $2$, we introduce the shorthand, $\bs^{(t)}$ for the supply $\bs_{\max}/2^t$.
We say that a column $j$ is in {\em class} $t\geq 0$,
if $\bs_j=\bs^{(t)}$, and we let 
$$ \C^{(t)} := \{ j \in \C\setminus F \,:\, \bs_j = \bs^{(t)} \} $$
be the set of class $t$ supplies.  \\

\vspace*{\parcor}
\paragraph{Step 2: Disregarding $i$-large columns of a small row $i$.}

Fix a small row $i \in \R_S$. We now identify the columns $j$ that are $i$-small.
To do so, define $t_i := \log (\bs_{max}/\bb_i) + 1$. Observe that any column $j$ in class $\C^{(t)}$ 
for $t\ge t_i$ are $i$-small. This is because $\bs_j = s_{max}/2^t \le s_{max}/2^{t_i} = \bb_i/2 < \bb_i$.
Define 
$$ \bb^{(t)}_i = \left\{\begin{array}{l@{\quad : \quad}l}
    2\sum_{j \in \C^{(t)}} A^F[\bs]_{ij}y_j  & t \geq t_i \\
    0 & \mbox{otherwise} \end{array}\right. $$
as the contribution of the class $t$, $i$-small columns to the demand of row $i$, multiplied by $2$.
Note that by definition of small rows, these columns contribute to more than $1/2$ of the demand, and so
\begin{equation}\label{eq:bbisum}
  \sum_{t \geq t_i} \bb^{(t)}_i \geq \bb_i.
\end{equation}
Henceforth, we will consider only the contributions of the small $i$-columns of a small row $i$.\\

\vspace*{\parcor}
\paragraph{Step 3: Scaling and getting the integral solution.}

Fix a class $t$ of columns and scale down by $\bs^{(t)}$ to get a $\{0,1\}$-constraint matrix.
(Recall entries of the columns in a class $t$ are all $\bs^{(t)}$.)
This will enable us to apply assumption $1$ and get a integral solution corresponding to these columns.
The final integral solution will be the concatenation of the integral solutions over the various classes.

The constants in the next claim are carefully chosen for the calculations to work out later.

\begin{claim} \label{claim:s1}
  For any $t\ge 0$ and for all $i \in \R_S$, $6\cdot \sum_{j \in \C^{(t)}} A_{ij}y_j \geq \lfloor 3\bb^{(t)}_i/\bs^{(t)} \rfloor$. 
\end{claim}
\begin{proof}
  The claim is trivially true for rows $i$ with $t_i > t$ as $\bb^{(t)}_i=0$ in this
  case. Consider a row $i$ with $t_i \leq t$.
  Since any column $j\in \C^{(t)}$ is $i$-small, we get $A^F[\bs]_{ij} = A_{ij}\bs_j = A_{ij}\bs^{(t)}$. 
  Using the definition of $\bb_i$, we
  obtain
  $$ 6 \cdot \sum_{j \in \C^{(t)}} A_{ij}\bs^{(t)} y_j
     = 3 \bb^{(t)}_i. $$
  Dividing both sides by $\bs^{(t)}$ and taking the floor on the right-hand side yields the claim.
\end{proof}

Since $\alpha=1/24$ and $\bx$ is a feasible solution to 
\cov{A^F[s],b^F,c,d/24}, we get that $6y_j = 24 \cdot \bx_j \leq
d_j$ for all $j \in \C\setminus F$. Thus, the above
claim shows that $6y$ is a feasible fractional solution for
\cov{A^{(t)},\lfloor 3\bb^{(t)}/\bs^{(t)}\rfloor,c^{(t)},d^{(t)}}, where 
$A^{(t)}$ is the submatrix of $A$ defined by the columns in $\C^{(t)}$, and 
$c^{(t)}$ and $d^{(t)}$ are the sub-vectors of $c$ and $d$, respectively, that are induced by
$\C^{(t)}$. Using Assumption \ref{as:1}, we therefore conclude that there is an integral vector $x^{{\tt int},\Sm,t}$
such that 

\begin{eqnarray}
  x^{{\tt int},\Sm,t}_j & \leq & d_j  \quad \mbox{for all $j \in \C^{(t)}$, and} \label{y-bd}\\
  \sum_{j \in \C^{(t)}} A^{(t)}_{ij} x^{{\tt int},\Sm,t}_j & \geq & \left\lfloor \frac{3\bb^{(t)}_i}{\bs^{(t)}}\right\rfloor  \quad \mbox{for all $i \in \R_S$, and} \label{hx-feas} \\
 \sum_{j\in \C^{(t)}} c_j  x^{{\tt int},\Sm,t}_j  & \leq & 6\gamma \cdot \sum_{j\in \C^{(t)}} c_j y_j
\end{eqnarray}
We obtain integral solution $\xints$ by letting $\xints_j=x^{{\tt int},\Sm,t}_j$ if $j \in \C^{(t)}$.
Thus $\xints_j \le d_j$ for all $j\in \C$, and we get, 
\begin{equation}\label{eq:cost-sm}
  \sum_{j \in \C} c_j \xints_j = \sum_{t \geq 0} \sum_{j \in \C^{(t)}} c_jx^{{\tt int},\Sm,t}_j   
  \leq 6\gamma \cdot \sum_{t \geq 0}\sum_{j \in \C^{(t)}} c_j y_j =  24\gamma \cdot \sum_{j \in \C} c_j \bx_j.
\end{equation}

\ni
Thus we have established parts (a) and (b) of the lemma.
It remains to show that $\xints$ is feasible for the set of small rows. \\

\vspace*{\parcor}
\paragraph{Step 4: Putting them all together: scaling back.}

Once again, fix a small row $i\in \R_S$. The following inequality takes only contribution of the $i$-small columns.
We later show this suffices.
\begin{align}\label{eq:feas}
\sum_{j\in \C}A^F[s]_{ij}\xints_j \ge ~~~~ \sum_{j\in \C: ~j \mbox{ is $i$-small}}A_{ij}s_j\xints_j ~~~~~ ~~~~~~~~~~~~~~~~~~~~~\notag \\ 
= ~~~ \sum_{t \geq t_i} \sum_{j \in \C^{(t)}} A^{(t)}_{ij}s_j \xints_j \geq \sum_{t \geq t_i} \sum_{j \in \C^{(t)}} A^{(t)}_{ij}\bs^{(t)} x^{{\tt int},\Sm,t}_j
\end{align}
The first inequality follows since $A^F[s]_{ij} = A_{ij}s_j$ for
$i$-small columns, the equality follows from the definition
of $t_i$, and the final inequality uses the fact that $s_j \geq
\bs^{(t)}$ for $j \in \C^{(t)}$. The following claim along with
\eqref{eq:feas} proves feasibility of row $i$.  This is the part where
our analysis slightly differs from the standard grouping-and-scaling
analysis.

\begin{claim}
For any small row $i\in \R_S$, 
$$ \sum_{t \geq t_i} \sum_{j \in \C^{(t)}} A^{(t)}_{ij}\bs^{(t)} x^{{\tt int},\Sm,t}_j \ge b^F_i.$$
\end{claim}
\begin{proof}
In this proof, the choice of the constant $3$ on the right-hand side
of the inequality in Claim \ref {claim:s1} will become clear. Let 
$ S_i = \{ t \geq t_i \,:\, 3\bb^{(t)}_i < \bs^{(t)} \}$
be the set of $i$-small classes $t$ whose fractional supply
$\bb^{(t)}_i$ is small compared to its integral supply $\bs^{(t)}$. 
We now show that for any small row $i$, the columns in the classes not in $S_i$ 
suffice to satisfy its demand. 
Note that
\begin{align}
  \sum_{t \not\in S_i, t \geq t_i} \bb^{(t)}_i = \sum_{t \geq t_i} \bb^{(t)}_i - \sum_{t \in S_i} \bb^{(t)}_i 
    \geq \sum_{t \geq t_i} \bb^{(t)}_i - \frac13 \sum_{t \in S_i} \bs^{(t)} \label{eq:1}
\end{align}
which follows from the definition of $S_i$. Furthermore, from \eqref{eq:bbisum} we know that for a small row,
$\sum_{t \geq t_i} \bb^{(t)}_i \geq \bb_i$. Also, since $\bs^{(t)}$ form a geometric series, we get that 
$\sum_{t \in S_i} \bs^{(t)} \le \sum_{t\ge t_i} \bs^{(t)} \le 2\bs^{(t_i)}$. Putting this in \eqref{eq:1} we get
\begin{align}
 \sum_{t \not\in S_i, t \geq t_i} \bb^{(t)}_i 
    \geq \bb_i - \frac13 \sum_{t \geq t_i} \bs^{(t)} 
    \geq \bb_i - \frac23 \bs^{(t_i)}  = \frac23 \bb_i,\label{eq:si} 
\end{align}
where the final equality follows from the definition of $t_i$ which implies that $\bs^{(t_i)}=\bb_i/2$. 

Moreover, for $t \not\in S_i$, we know that $\lfloor
3\bb^t_i/\bs^{(t)} \rfloor \geq \frac32 \bb^t_i/\bs^{(t)}$ since
$\lfloor a \rfloor \ge a/2$ if $a>1$.  Therefore, using inequality
\eqref{hx-feas} in \eqref{eq:feas}, we get
\begin{eqnarray*}
\sum_{j\in \C}A^F[s]_{ij}\xints_j \ge  \sum_{t \geq t_i} \sum_{j \in \C^{(t)}} A^{(t)}_{ij}\bs^{(t)} x^{{\tt int},\Sm,t}_j & \geq & \sum_{t \not\in S_i, t \geq t_i} \bs^{(t)} \left\lfloor \frac{3\bb^{(t)}_i}{\bs^{(t)}} \right\rfloor\\
      & \geq & \frac32 \sum_{t \not\in S_i, t \geq t_i} \bb^{(t)}_i \\
      & \geq & \bb_i \ge b^F_i,
\end{eqnarray*}
where the second-last inequality uses \eqref{eq:si}, and the last uses the definition of $\bb_i$. This completes the proof of the lemma.
\end{proof} }
\end{proof}

\subsubsection{Large rows.}
The large rows can be showed to be a PCIP problem and thus Assumption 2 can be invoked to get
an analogous lemma to Lemma \ref{lem:small-rows}.
\begin{lemma}\label{lem:large-rows}
We can find an integral solution $\xintl$ such that \\
\indent
a) $\xintl_j \le 1$ for all $j$, \\ \indent
b) $\sum_{j\in \C}c_j\xints_j \le 8\omega \sum_{j\in \C}c_j\bx_j$, and \\ \indent
c) for every  large row $i\in \R_L$, $\sum_{j\in \C} A^F[s]_{ij}\xints_j \ge b^F_i$.
\end{lemma}
\iffull{
\begin{proof}
Let $i \in \R_L$ be a large row, and recall that $\La_i$ is the set of $i$-large columns in \C. We have
$$ \sum_{j\in \La_i}A^F[s]_{ij}y_j = \sum_{j \in \La_i} A_{ij} \bb_i y_j \geq \bb_i/2, $$
and hence 
\begin{equation}\label{eq:trc} 
  2\sum_{j \in \La_i} A_{ij} y_j \geq 1. 
\end{equation}

Let $A^\R$ be the minor of $A$ induced by the large rows. Consider the
priority cover problem \cov{A^\R[\bs,\bb],\mathds{1},c}. From the
definition of $\La_i$, it follows $2y$ is a feasible fractional
solution to the priority cover problem.

Using Assumption \ref{as:2}, we conclude that there is an 
integral solution $\xintl$ such that $\sum_{j\in \C}c_j\xintl_j \le 2\omega \sum_{j \in \C} c_jy_j = 8\omega \sum_{j \in \C} c_j \bx_j$,
and  $\sum_{j\in \C}A^\R_{ij}\xintl_{j} \ge 1$, for all large rows $i\in \R_L$.

Fix a large row $i$.  Since $A^F[s]_{ij} = b^F_i$  for all $i$-large columns $\La_i$, we get
$$\sum_{j\in \C}A^F[s]_{ij}\xintl_j \ge \sum_{j\in \La_i} A_{ij}b^F_i\xintl_{j} = b^F_i \sum_{j\in C} A^\R_{ij}\xintl_j \ge b^F_i$$
This completes the proof of the lemma.
\end{proof}

\ni
{\bf Proof of Theorem \ref{thm:1}}
Let $\xints$ and $\xintl$ be as satisfying the conditions of Lemma \ref{lem:small-rows} and \ref{lem:large-rows}, respectively. Define $\xint$ as
$\xint_j = \max\{\xints_j , \xintl_j\}$. We have\\

\noindent
a) $\xint_j \le d_j$ since both $\xints_j \le d_j$ and $\xintl_j \le 1 \le d_j$. \\

\noindent
b) For any row $i$, $\sum_{j\in \C} A^F[s]_{ij}\xint_j \ge b^F_i$ since the inequality is true with $\xint$ replaced by $\xints$ for small rows, and 
$\xint$ by $\xintl$ for large rows. \\

\noindent
c) $\sum_{j\in \C} c_j\xint_j \le \sum_{j\in \C}c_j\xints_j + \sum_{j\in\C} c_j\xintl_j \le (24\gamma+8\omega)\sum_{j\in \C} c_j\bx_j$. \\

\ni
Thus, $\xint$ is a feasible integral solution to $\cov{A^F[s],b^F,c,d}$ with cost bounded as $\sum_{j\in \C} c_j\xint_j \le (24\gamma+8\omega)\sum_{j\in \C} c_j\bx_j$. Noting that $\alpha=1/24$, the proof of the theorem follows from Lemma \ref{lem:res}. $\square$.
}
\ifconf{
Define  $\xint$ as $\xint_j = \max\{\xints_j , \xintl_j\}$ for all $j$; using the previous two lemmas and Lemma \ref{lem:res}, this integral solution proves Theorem \ref{thm:1}.
}

\iffull{
\subsection{CCIPs with violation of upper-bounds: Proof of Theorem \ref{thm:2}}
In this section we prove Theorem \ref{thm:2} that we restate here. In the proof, we will indicate how
we modify the analysis of grouping-and-scaling that allows us to replace the no-bottleneck assumption
with a weaker one.
\begin{theorem}(Theorem \ref{thm:2})
Under assumption \ref{as:1} and assuming $A_{ij} s_j \le b_i$, for all $i,j$, given a fractional solution $x$ to 
the canonical LP relaxation of \cov{A[s],b,c,d}, one can find an integral solution $\xint$ whose cost 
$c\cdot \xint \le 10\gamma (c\cdot x)$ and $\xint \le 10d$.
\end{theorem}
\begin{proof}
\def\xint{x^{\tt int}}
\def\tplus{2^{-(t+1)}s_{max}}
\def\t{2^{-t}s_{max}}
\def\Aij{A_{ij}}
Let $x$ be a feasible solution to $A[s]x \ge b, x\ge 0$. We construct an integral solution $\xint$ such that $A[s]\xint \ge b$ and $c^T\xint \le 10\gamma c^Tx$.
Let $s_{max}$ and $s_{min}$ be the largest and smallest $s_j$'s.\\

\noindent
{\bf \em Grouping:}
Let $\C^{(t)} := \{j: \tplus ~ < s_j \le ~~ \t\}$ for $t=0,1,\ldots T$
where $T = \log(\frac{s_{max}}{s_{min}})$.  Let $b^{t}_i :=
\sum_{j \in \C^{(t)}} \Aij s_j x_j$. Note that $\sum_{t=0}^T b^t_i \ge
b_i$. Let $m^t_i := \min_{j\in \C^{(t)}: \Aij\neq 0} s_j\Aij$, that
is, $m^t_i$ is the smallest non-zero entry of the $i$th row of $A$ in
the columns of $\C^{(t)}$.  Note that $m^t_i > \tplus$. Let $m_i$ be
the largest entry of row $i$. The assumption $A_{ij}s_j \le b_i$
implies
$m_i \le b_i$. \\

\ni
\def\s{\hat{s}}
{\bf \em Scaling:}
Let $y^t$ be a vector with $y^t_j = 10x_j$ for $j\in \C^{(t)}$, $0$ elsewhere. Note that $\sum_t c^T y^t = 10c^T x$
and $y^t_i\le 10d_i$ for any $i$.
Let $\s^t$ be a vector with $\s^t_j = \tplus$ for $j\in \C^{(t)}$, $0$ otherwise. Since for all $j\in \C^{(t)}$, $\s^t_j \ge s_j/2$, for all rows $i$ we have 
$$\sum_{j\in \C^{(t)}} \Aij \s^t_jy^t_j \ge 5 \sum_{j\in \C^{(t)}} \Aij s_jx_j = 5b^t_i$$
\ni
Therefore since $m^t_i \ge \tplus$, we get
\begin{align*} \label{eq:ineq1} \sum_{j\in \C^{(t)}} \Aij y^t_j \ge
  \frac{5b^t_i}{\tplus} \ge \frac{5b^t_i}{m^t_i} \ge \left\lfloor
    \frac{5b^t_i}{m^t_i} \right\rfloor
\end{align*}
\ni If we define an integral vector $a^t$ to be $a^t_i := \lfloor
\frac{5b^t_i}{m^t_i} \rfloor$, we see that $Ay^t \ge a^t$.
Using assumption \ref{as:1},  there exists an integral solution $z^t$ such that $Az^t \ge a^t$, and $c^T z^t \le \gamma(c^T y^t)$, and $z^t_i \le 10d_i$. \\

\ni
{\bf \em Scaling back:}  Now fix a row $i$, and look at 
$$\sum_{j\in \C^{(t)}} \Aij s_j z^t_j \ge \sum_{j\in \C^{(t)}} \Aij m^t_i z^t_j = m^t_i \sum_{j\in \C^{(t)}} \Aij z^t_j \ge m^t_i \lfloor  \frac{5b^t_i}{m^t_i} \rfloor$$
where the first inequality follows 
since $m^t_i$ is the minimum entry in the $i$th row in the columns of $\C^{(t)}$. This is where our analysis slightly differs
from the previous analyses of grouping and scaling, where instead of multiplying the RHS by $m^t_i$, the RHS was multiplied
by $2^{-t}s_{max}$. This subtle observation leads us to make a weaker assumption than the no-bottleneck assumption. \\

\def\zint{z^{\bf int}}
\ni
{\bf \em Getting the final integral solution:} \\
Define $\xint := \sum_{t=0}^T z^t$. Note that $c^T \xint = \sum_t c^T z^t \le \gamma \sum_t c^T y^t = 10\gamma (c^T x)$
and $\xint \le 10d$. \\

\ni
Fix a row $i$ and look at the $i$th entry of $A[s]\xint$.
\begin{equation}\label{eq:app1}
\sum_{t=0}^T\sum_{j\in \C^{(t)}}  \Aij s_j z^t_j \ge  \sum_{t=0}^T  \lfloor  \frac{5b^t_i}{m^t_i} \rfloor m^t_i
\end{equation}
\ni
Let $S_i := \{t: 5b^t_i < m^t_i\}$. Note that 
$$\sum_{t\in S_i} b^t_i < \frac{1}{5} \sum_{t\in S_i} m^t_i \le 3m_i/5$$
the second inequality following from Claim \ref{claim:m} below. This gives us
$$\sum_{t\notin S_i} b^t_i > \sum_{t=0}^T b^t_i - 3m_i/5 \ge b_i - 3m_i/5$$

For $t\notin S_i$, we have the floor in the inequality \eqref{eq:app1} at least 1. So we can use the relation $\lfloor x \rfloor \ge x/2$ for $x\ge 1$. Thus, using $m_i \le b_i$,
we have

$$ A\xint \ge \sum_{t\notin S_i} \frac{5b^t_i}{2} \ge 5b_i/2 - 3m_i/2 \ge b_i$$

\begin{claim}\label{claim:m}
$\sum_{t=0}^T m^t_i \le 3m_i$. 
\end{claim}
\begin{proof}
Note that the non-zero $m^t_i$ decreases as $t$ goes from $0$ to $T$. Also, for any $t < t'$, we have $m^t_i > \tplus$ and $m^{t'}_i \le 2^{-t'}s_{max}$. Thus, $m^{t'}_i \le m^t_i \cdot 2^{-(t' - t - 1)}$. Since the largest $m^t_i$ can be at most $m_i$, 
$\sum_{t=0}^T m^t_i \le m_i + m_i + m_i/2 + m_i/4 + .... \le 3m_i$. 
\end{proof}\end{proof}
}
\section{Priority line cover}
\label{sec:plc}

\ifconf{ In this extended abstract, we show that the integrality gap
  of the canonical linear programming relaxation of PLC is at most
  $2$.  Subsequently, we sketch an exact combinatorial algorithm for
  the problem.  } 

\iffull{ We first show that the integrality gap of
  the canonical linear programming relaxation of PLC is at least $3/2$
  and at most $2$. Subsequently, we present an exact combinatorial
  algorithm for the problem.  }

\subsection{Canonical LP relaxation: Integrality gap}
We start with the canonical LP relaxation for PLC and its dual in Figure \ref{fig:PD}.
\begin{figure}[h]
  \begin{minipage}{\halftw} \begin{align}
      \min \Big\{\sum_{j \in \ess} c_jx_j: \quad& x \in R^\ess_+
\tag{Primal}\label{lp:primal} \\
      \sum_{j \in \ess: j\,\textrm{\scriptsize  covers } e} x_j \geq 1, \quad& \forall e
\in E \notag \Big\}
    \end{align} \end{minipage}
  \hfill \vline \hfill
  \begin{minipage}{\halftw} \begin{align}
      \max \Big\{\sum_{e\in E}  y_e: \quad& y \in R^E_+ \tag{Dual}\label{lp:dual} \\
      \sum_{e\in E: j \,\textrm{\scriptsize  covers } e} y_e \leq c_j,
\quad&\forall j\in \ess  \notag \Big\}
    \end{align}\end{minipage}
  \caption{\small The PLC canonical LP relaxation and its
dual.}\label{fig:PD} \end{figure}

\iffull{
The following example shows that the integrality gap of
\eqref{lp:primal} is at least $3/2$.

\begin{example}
  {\em Figure \ref{fig:1} shows a line of odd length $k$; odd numbered
  edges have demand $1$, and even numbered edges have a demand of
  $2$. Paths are shown as lines above the line graph, and are also
  numbered. Odd numbered paths have a supply of $2$, and even numbered
  ones have a supply of $1$. Dashed lines indicate edges spanned but
  not covered.  All paths have cost $1$. Note that a fractional
  solution is obtained by letting $x_p=2/3$ for paths $2$ and $k$, and
  $x_p=1/3$ otherwise.  The cost of this solution is $(k+3)/3$, while
  the best integral solutions takes all odd-numbered paths, and has
  cost $(k+1)/2$. As $k$ tends to $\infty$, the ratio between the
  integral and fractional optimum tends to $3/2$. As an aside, we
  found the above integrality gap instance by translating a known
  integrality-gap instance of the tree-augmentation problem in
  caterpillar graphs; see \cite{CK+08}.}

 \begin{figure}[h]
    \begin{center}
      \includegraphics[scale=0.8]{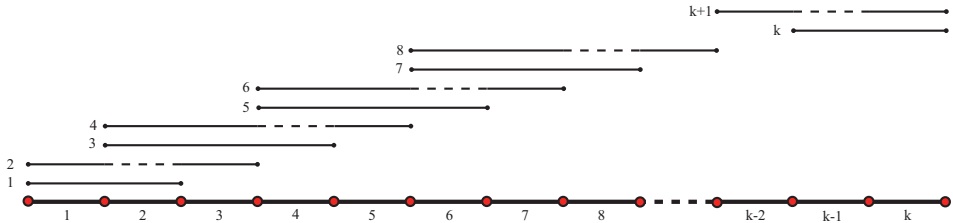}
      \label{fig:1}  \caption{Integrality Gap for PLC}
    \end{center}
  \end{figure}
\end{example}

\noindent
We now show that the integrality gap of the canonical LP for PLC is bounded by $2$. 
We describe a simple primal-dual algorithm that constructs a
feasible line cover solution and a feasible dual solution, and the
cost of the former is at most twice the value of the dual solution. 

The algorithm maintains a set of segments $Q$.  Call an edge $e$ {\em
  unsatisfied} if no segment in $Q$ covers $e$. Let $U$ be the set of
unsatisfied edges.  Initially $Q$ is the empty set and $U = E$.  We
grow duals $y_e$ on certain edges, as specified below. We let $E_+$
denote the edges with positive $y_e$; we call such edges, {\em
  positive} edges. Initially $E_+$ is empty.  Call a segment $j$ {\em
  tight} if $\sum_{e\in j: j~covers~e} y_e = c_j$.
We use the terminology an edge $e$ is larger than $f$, if $\de_e \ge \de_f$. \\

\vspace{2mm}
\hspace{-6mm}
\begin{boxedminipage}{\textwidth}
{\bf Primal-Dual Algorithm} 
\begin{enumerate}
\item While $U$ is not empty do
\begin{itemize}
	\item Breaking ties arbitrarily, pick the largest edge $e$ in $U$.
	\item Increase $y_e$ till some segment becomes tight. Note
          that each such segment must contain $e$. Let $j_l(e)$ and
          $j_r(e)$ be the tight segments that have the smallest
          left-end-point and the largest right-end-point,
          respectively. Since $e$ is chosen to be the largest
          uncovered edge, any unsatisfied edge contained in 
          the two segments $j_l(e)$
          or $j_r(e)$ is also covered.  We say $e$
          is responsible for $j_l(e)$ and $j_r(e)$.
	
	Add $j_l(e),j_r(e)$ to $Q$. Add $e$ to $E_+$.
	Remove all the unsatisfied edges contained in either $j_l(e)$ or $j_r(e)$ from $U$.
\end{itemize}

\item {\bf Reverse Delete:} Scan the segments $j$ in $Q$ in the
  reverse order in which they were added, and delete $j$ if its
  deletion doesn't lead to uncovered edges.
\end{enumerate}
\end{boxedminipage}

\vspace{4mm}
\noindent
It is clear that the final set $Q$ is feasible. It is also clear that
$y$ forms a feasible dual.  The factor $2$-approximation follows from
the following lemma by a standard relaxed complementary slackness
argument, and this finishes the proof of Theorem \ref{thm:plc-gap}.
\begin{lemma}\label{lem:atmost2}
  Any edge $e \in E_+$ is covered by at most two segments in $Q$.
\end{lemma}

\begin{proof}
  Suppose there is an edge $e\in E_+$ covered by three segments
  $j_1,j_2$ and $j_3$.  Observe that one of the segments, say $j_2$,
  must be completely contained in $j_1\cup j_3$.  Since $j_2$ is not
  deleted from $Q$, there must be an edge $f\in j_2$ such that $j_2$
  is the only segment in $Q$ covering $f$. Since $j_1$ and $j_3$ don't
  cover $f$, but one of them, say $j_1$ contains it, this implies
  $\de_f > \sup_{j_1} \ge \de_e$. That is, $f$ is larger than $e$.

  If $f$ is the edge responsible for $j_2$, then since $j_2$ contains
  $e$, $e$ wouldn't be in $E_+$. Since $f$ is larger than $e$, there
  must be a segment $j$ in $Q$ added before $j_2$ that covers $f$. In
  the reverse delete order, $j_2$ is processed before $j$.  This
  contradicts that $j_2$ is the only segment in $Q$ covering $f$.
\end{proof}

\begin{lemma}
  $\sum_{j\in Q} c_j \le 2\sum_{e\in E} y_e$.
\end{lemma}
\begin{proof}
  Since each $s\in Q$ satisfies $\sum_{e\in j: j~covers~e} y_e = c_j$, we get
  $$\sum_{j\in Q} c_j  = \sum_{j\in Q} \sum_{e\in j: j~covers~e} y_e = 
  \sum_{e\in E} y_e \cdot |\{j\in Q: j~covers~e\}| \le 2\sum_{e\in E} y_e$$
\end{proof}
}
\ifconf{
We use the terminology an edge $e$ is larger than $f$, if $\pi_e \ge \pi_f$. 
The algorithm maintains a set of segments $Q$ initially empty.  Call an edge $e$ {\em
  unsatisfied} if no segment in $Q$ covers $e$ and let $U$ be the set of
unsatisfied edges. The algorithm picks the largest edge in $U$ and raises the dual 
value $y_e$ till some segments becomes tight. The segments with the farthest left-end point
and the farthest right-end point are picked in $Q$, and all edges contained in any of them 
are removed from $U$. Note that since we choose the largest in $U$, all such edges are covered.
The algorithm repeats this process till $U$ becomes $\emptyset$, that is,
all edges are covered. The final set of segments is obtained by a reverse delete step, where a segment
is deleted if its deletion doesn't make any edge uncovered.

The algorithm is a factor $2$ approximation algorithm. To show this it suffices 
by a standard argument for analysing primal-dual algorithms, that any edge 
with a positive dual $y_e$ is contained in at most two segments in $Q$. 
These two segments correspond to the left-most and the right-most segments 
that cover $e$; it is not too hard to show if something else covers $e$, then either
$e$ has zero dual, or the third segment is removed in the reverse delete step.
}
\subsection{An Exact Algorithm for PLC}
\iffull{
We first describe the sketch of the algorithm; the full proof starts
from Section \ref{sec:vms}.}
\ifconf{We sketch the exact algorithm for PLC.}
 A segment $j$ covers only a subset of
edges it contains. We call a contiguous interval of edges covered by
$j$, a {\em valley} of $j$. The uncovered edges form {\em
  mountains}. Thus a segment can be thought of as forming a series of
valleys and mountains.

Given a solution $S\subseteq \ess$ to the PLC (or even a PTC)
instance, we say that segment $j\in S$ is {\em needed} for edge $e$ if
$j$ is the unique segment in $S$ that covers $e$. We let $E_{S,j}$
be the set of edges that need segment $j$. 
We say a solution is {\em
  valley-minimal} if it satisfies the following two properties: (a) If
a segment $j$ is needed for edge $e$ that lies in the valley $v$ of
$j$, then no higher supply segment of $S$ intersects this valley $v$,
and (b) every segment $j$ is needed for its last and first edges.  We
show that an optimum solution can be assumed to be valley-minimal, and
thus it suffices to find the minimum cost valley-minimal solution.

The crucial observation follows from properties (a) and (b) above.
The valley-minimality of solution $S$ implies that there is a unique
segment $j \in S$ that covers the first edge of the line. At a very
high level, we may now use $j$ to decompose the given instance into a
set of {\em smaller} instances.  For this we first observe that each
of the remaining segments in $S \setminus \{j\}$ is either fully
contained in the strict interior of segment $j$, or it is disjoint
from $j$, and lies to the right of it.  The set of all segments that
are disjoint from $j$ form a feasible solution for the smaller PLC
instance induced by the portion of the original line instance to the
right of $j$. On the other hand, we show how to reduce the problem of
finding an optimal solution for the part of the line contained in $j$
to a single shortest-path computation in an auxiliary digraph.  Each
of the arcs in this digraph once again corresponds to a smaller
sub-instance of the original PLC instance, and its cost is that of its
optimal solution. The algorithm follows by dynamic programming.
\iffull{
\subsubsection{Valley-Minimal Solutions}\label{sec:vms}
As mentioned above, it helps to think of supplies and demands as {\em
  heights}. In the case of PLC, the demands of the edges in $E$ form a
terrain, and each segment $j \in \ess$ corresponds to a straight line
at height $\su_j$. Segment $j$ then covers edge $e$ if $e$ lies in the
segment's {\em shadow}, that is, the height of $e$ is smaller than the
height of the segment.

\begin{figure}[h]
\begin{center}
 \includegraphics[scale=0.8]{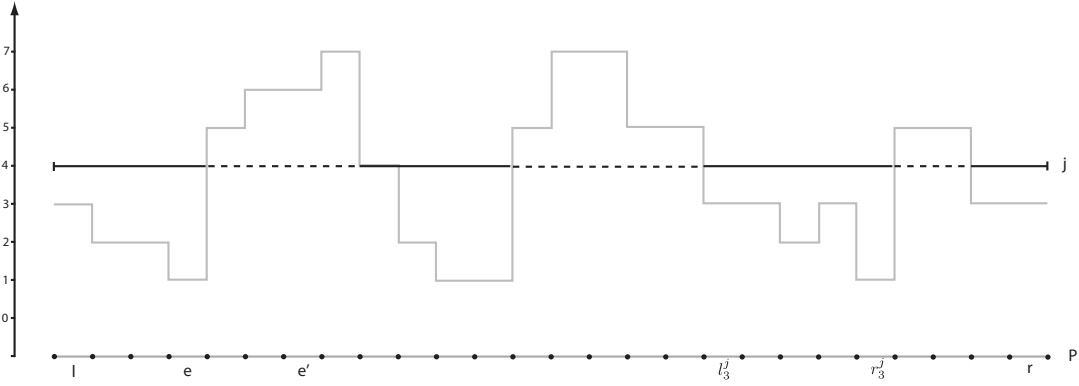}
\end{center}
\caption{\label{fig:terr} The figure shows a segment $j$, and the terrain
  induced by the edges of $E$ that it contains. The terrain partitions
  $j$ into valleys and mountains. Valleys are indicated by solid parts of
  $j$, and mountains are shown as dashed lines.}
\end{figure}

Figure \ref{fig:terr} illustrates this with path $P$ and its edges. The light
gray terrain indicates the demands of the edges. The segment $j$ shown
in the picture covers the edges in $[l,r]$ that lie in its shadow;
e.g., $j$ covers edge $e$ but not $e'$. The terrain partitions
$j$ naturally into {\em valleys} -- contiguous sub-intervals of
$[l,r]$ that are in the shadow of $j$, and {\em mountains} -- those
sub-intervals that are contained in $[l,r]$ and consist entirely of
edges that are not covered by $j$. The parts of $j$ that correspond to
mountains are indicated by dashed lines, and valleys are depicted by
solid lines. In the following, we let $[l^j_k,r^j_k]$ be the interval
corresponding to the $k$th valley of $j$.

In the following, we will assume that the set of segments \ess\ in the
given PLC/PTC instance is {\em segment-complete}; i.e., if \ess\
contains the segment $j$ then it also contains all proper
sub-segments. For example, if a PLC instance contains segment $j$
corresponding to interval $[l^j,r^j]$, then it also contains segments
corresponding to intervals $[l,r]$ for all $l^j \leq l \leq r \leq
r^j$. This assumption is w.l.o.g. as we can always add a {\em dummy}
sub-segment $j'$ for any such interval $[l,r]$ with the same supply
and cost as $j$. Any minimal solution clearly uses at most one of $j$
and $j'$, and if $j'$ is used, then replacing it with $j$ does not
affect feasibility.

Let $S \subset \ess$ be an inclusion-wise minimal solution for the given 
instance, and let $j \in S$ be any one of its segments. We say that $j$ is
{\em needed} for edge $e \in E$ if $j$ covers $e$, and if there is no
other segment in $S$ that covers $e$; let $E_{S,j}$ be the set of edges that 
need $j$, and hence $E_{S,j}\neq \varnothing$ for all $j \in S$. Thus, if
$j$ is needed for $e$, then $e$ is in one of $j$'s valleys; we let $\val^j_e$
be that valley. 

A solution $S \subseteq \ess$ is {\em
  valley-minimal} if 
\begin{itemize}
\item[ {[M1]} ] for all $j \in S$ and for all $e \in E_{S,j}$, no
  segment of higher supply in $S$ covers any of the edges in
  $\val^j_e$, and
  \item[ {[M2]} ]  each segment is needed for its first and 
    last edge.
\end{itemize}

We obtain the following observation.

\begin{lemma}
  Given a feasible instance of PLC/PTC, there exists an optimum feasible
  solution that is valley-minimal.
\end{lemma}
\begin{proof}
  First, it is not too hard to see that we can always obtain an optimal 
  solution that satisfies [M2]. If $S$ is an optimum solution with a 
  segment $j$, and $j$ is not needed for its first or last edge $e$,
  then we may clearly replace $j$ by the sub-segment $j-e$. This 
  does not increase the solutions cost, using the segment-completeness.
  
  Assume, for the sake of contradiction that $S$ violates [M1]. 
For a solution $S \subseteq \ess$, say that 
$(j,j',e)$ is a {\em violating triple} if $j,j' \in S$, $j'$ has higher supply than  $j$, $j$ is
needed for $e$, and $j'$ covers some edge in $\val^j_e$. 
  Choose a solution $S$ with the smallest number of violating triples and
  let $(j,j',e)$ be one such triple. Since $j$ is needed for $e$, edge
  $e$ is not contained in $j'$, and hence $j'$ is either fully
  contained in the interval $(e,n]$ or fully contained in the interval
  $[1,e)$.  Using the segment-completeness assumption, we may replace
  $j'$ by the sub-segment $j''$ obtained by removing the prefix
  consisting of edges in $\val^j_e$; remove $j''$ if it is empty. The
  resulting set of segments has cost at most that of $S$, and the
  number of violating triples is smaller; a contradiction.
\end{proof}

In the next subsection, we show how we can compute the minimum cost
valley-minimal solution for PLC instances in polynomial time using
dynamic programming.

\subsubsection{Computing valley-minimal solutions}

Given $1 \leq l \leq r \leq n$, we obtain the sub-instance {\em
  induced} by interval $[l,r]$ by restricting the line $[1,n]$ to this
interval, and by keeping only segments that are fully contained in
$[l,r]$.  Observe that the valley-completeness assumption implies that
any such sub-instance is feasible.  We begin by making a crucial
observation that will allow us to decompose a given PLC instance into
{\em smaller} instances.  Let $S$ be a valley-minimal solution for the
sub-instance induced by $[l,r]$, and note that [M2] implies that $S$
contains a unique segment $j$ that covers the first edge $(l,l+1)$.
Suppose that $E_{S,j} = \{e_1, \ldots, e_k\}$ is the set of edges
within $[l,r]$ that need segment $j$. Abusing notation slightly, we
let $\val^j_i=[l^j_i,r^j_i]$ be the valley of $j$ around edge $e_i$;
thus we clearly have
\begin{equation}\label{eq:vdec}
  E_{S,j} \subseteq \val^j_{1} \cup \ldots \cup \val^j_{k}.
\end{equation}
Note that segment $j$ may have valleys that entirely consist of 
edges that do not need $j$; accordingly, such valleys are not 
part of the list on the right-hand side of \eqref{eq:vdec}. 
Using property [M2], however, we may assume that
 $\val^j_1$ and $\val^j_k$ are the first and last valley, 
respectively, of segment $j$.
We obtain the following observation,  where we let $l^j_{k+1}=r+1$. 

\begin{observation}\label{obs:part}
  We may assume,
  for all $1 \leq i \leq k$, if $j' \in S$ contains $e \in (r^j_i,l^j_{i+1})$,
  then $j'$ is fully contained in $(r^j_i,l^j_{i+1})$. 
\end{observation}
\begin{proof}
  Consider first a segment $j' \in S$ with supply bigger than $s_j$.
  In this case [M1] implies that $j'$ must have an empty intersection with
  the valleys $\val^j_1, \ldots, \val^j_k$, and the observation
  follows.

  On the other hand if segment $j'$ has supply at most $s_j$, then
  since $j'$ must be needed for some edge $e$, $j$ must not contain
  $e$ implying $j'$ must have its right end-point in
  $(r^j_k,r]$. Replacing $j'$ by its intersection with $(r^j_k,r]$
  completes the observation.
\end{proof}

We now let $\OPT_{l,r}$ be a minimum cost valley-minimal feasible
solution for the sub-instance induced by interval $[l,r]$, and we let
$\opt_{l,r}$ be its cost. Clearly, $\OPT_{n,n}$ consists of the
minimum cost segment in \ess\ that covers edge $n$, and $\OPT_{1,n}$
is the optimum solution we want to obtain. Suppose that we know
$\OPT_{l',r'}$ for all $l < l' \le r' < r$. The high level idea is the
following. The algorithm guesses the first segment $j$ in
$\OPT_{l,r}$. Suppose that $r' \leq r$ is the rightmost edge covered
by $j$.  Observation \ref{obs:part} allows us to partition the
remaining segments in $\OPT_{l,r}$ into two parts:
\begin{description}
  \item[Part 1] Segments that contain edges in $(r',r]$. None of these 
    segments can contain any of the edges in $[l,r']$ by the observation.
  \item[Part 2] Segments that contain edges in $[l,r']$. Once again,
    the observation implies that such segments must be fully
    contained in $(l,r')$.
\end{description}
The first part's solution is obtained since it
is a smaller subproblem, the second part is obtained via a
shortest-path computation.  We now elaborate and give the complete
algorithm.

\begin{figure}[t]
\begin{center}
  \includegraphics[scale=0.8]{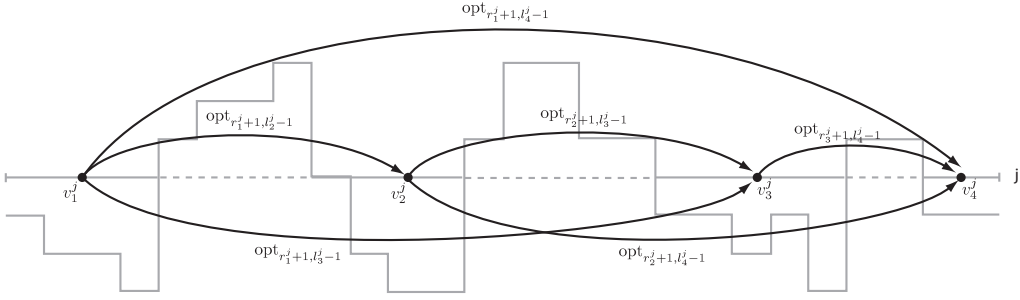}
\end{center}
\caption{\label{fig:gi} The part of digraph $G_l$ corresponding to segment
  $j \in \ess_l$.}
\end{figure}

Let $\ess_l$ be the segments in \ess\ with leftmost endpoint $l$.  We
construct a digraph $G_\ell$ as follows.  Consider a segment $j \in
\ess_l$, and let
$$ [l^j_1,r^j_1], \ldots, [l^j_k,r^j_k] $$
be the set of its valleys. We add a node $v^j_q$ for each valley $1
\leq q \leq k$ of $j$ to $G_\ell$. We also add an arc
$(v^j_q,v^j_{q'})$ for all $1 \leq q < q' \leq k$.
A shortest path corresponding to the solution $\OPT_{l,r}$ will use arc
$(v^j_q,v^j_{q'})$ if 
\begin{itemize}
\item[(i)] $j$ is the leftmost segment in $\OPT_{l,r}$, and
\item[(ii)] $\val^j_q$ and $\val^j_{q'}$ are two consecutive valleys 
  of $j$ that contain edges that need $j$. 
\end{itemize}
Observation \ref{obs:part} then states that $\OPT_{l,r}$ uses segments 
that are entirely contained in $(r^j_{q},l^j_{q'})$ to cover $(r^j_q,l^j_{q'})$.
An optimum set of such segments is given by $\OPT_{r^j_q+1,l^j_{q'}-1}$,
and we therefore give arc $(v^j_q,v^j_{q'})$ cost $\opt_{r^j_q+1,l^j_{q'}-1}$.
Figure \ref{fig:gi} shows the part of $G_\ell$ for the segment $s$ from 
Figure \ref{fig:terr}.

We add a source node $s_l$ and arcs $(s_l,v^j_1)$ of cost
$c_j$ for each of the segments $j \in \ess_l$. A shortest path uses
such an arc if $j$ is the unique segment starting at $l$ in the
corresponding optimum solution. 
We also add a sink node $t_r$ and add an arc $(v^j_k,t_r)$ for all $j\in \ess_l$
of cost $\opt_{r^j_k + 1,r}$ indicating the optimum PLC for the sub-interval $[r^j_k+1,r]$.
Note that if $r^j_k = r$, then this arc is a loop of cost $0$ and can be discarded.

It follows from the above construction that $\opt_{l,r}$ is equal to
the cost of a shortest $s_l,t_r$-path in $G_l$.  Each of the
shortest-path computations can clearly be done in polynomial time, and
hence $\opt_{l,r}$ can be obtained via dynamic programming, in
polynomial time.  This yields the following restatement of 
Theorem \ref{thm:plc-exact}.

\begin{theorem}
  The cost $\opt_{1,n}$ of an optimum solution for a given PLC instance
  can be computed in polynomial time.
\end{theorem}
}
\section{Priority tree cover}
\label{sec:ptc}
\ifconf{
In this extended abstract, we sketch a factor $2$ approximation for the PTC problem, and show how the PTC problem is a special case of the $3$ dimensional 
rectangle cover problem. For the APX hardness and the integrality gap of the unweighted PTC LP, we refer the reader to the full version.
}
\iffull{
We first give a proof of Theorem \ref{thm:ptc-hard}, and show
that rooted PTC is APX-hard, even if all segments have unit
cost. Subsequently, we present a $2$-approximation algorithm for the
problem, by reducing it to an auxiliary instance of the tree
augmentation problem. Then, we prove Theorem \ref{thm:ptc-gap},
and show that the integrality gap of the canonical LP formulation 
of unweighted PTC is bounded by $6$. Finally, we prove the connection
between PTC and the rectangle cover problem.

\subsection{APX-hardness}
We prove APX-hardness of PTC via a reduction from the minimum vertex cover problem in bounded degree graphs. The latter problem is known to be APX-hard \cite{BK98}.
Given a bounded degree graph $G(V,E)$,
with $n$ vertices and $m=O(n)$ edges, let the edges be arbitrarily numbered $\{1,2,\ldots,m\}$. 

The tree in our instance has a broom structure: it has a {\em handle} which is a path of $m$ edges $(e_1,\ldots,e_m)$ given by vertices 
$\{x_0,x_1,\ldots,x_m\}$, and it has $n$ {\em bristles} where each bristle corresponds to a particular vertex $v\in V$ and is a path of length $deg(v)$. The edge $e_i$ in the handle for $1\le i\le m$, corresponds to the edge numbered $i$ in the graph $G$.
The bristle corresponding to vertex $v$ is a path $(f^v_1,f^v_2,\ldots,f^v_{deg(v)})$ given by the vertices $\{x_m,y^v_1,y^v_2,\ldots,y^v_{deg(v)}\}$.  
The root of the tree is $x_0$, the end point of the handle.
Thus the tree has $m+ \sum_v deg(v) = 3m$ edges.

We now describe the priority demands of these tree edges.
The demand of edge $e_i$ is $i$.
Consider the edges in $G$ incident on $v$ in the decreasing order of their numbers. Suppose they are $(i_1 > i_2 > \cdots > i_{deg(v)})$. 
The demands of the edge $f^v_j$ is $i_j$. Thus, for a particular bristle corresponding to a vertex $v$, the demands decrease as we go from $f^v_1$ to $f^v_{deg(v)}$, and these demands correspond to the numbers of edges incident on $v$.

Now we describe the segments. All segments have unit cost.
We have two kinds of segments: edge segments and vertex segments. For every edge $i=(v,w)$ in $E$, there are two edge segments $s^i_v$ and $s^i_w$. Segments $s^i_v$ 
contains all edges $e_i$ to $e_m$ and edges $f^v_1$ to $f^v_j$, 
where edge $i$ is the $j$th edge in the descending order of neighbors of $v$ in $G$. The supply of segment $s^i_v$ is $i$, and thus by construction, we see that $s^i_v$ only spans edge $e_i$ and $f^v_j$. That completes the description of edge segments.
For every vertex $v$, there is a vertex segment $t_v$ that covers all the edges in the bristle corresponding to vertex $v$. 
That completes the description of the PTC instance. 
Look at figure \ref{fig:red} for an illustration of the reduction.

\begin{figure}
  \begin{center}
    \includegraphics[scale=.8]{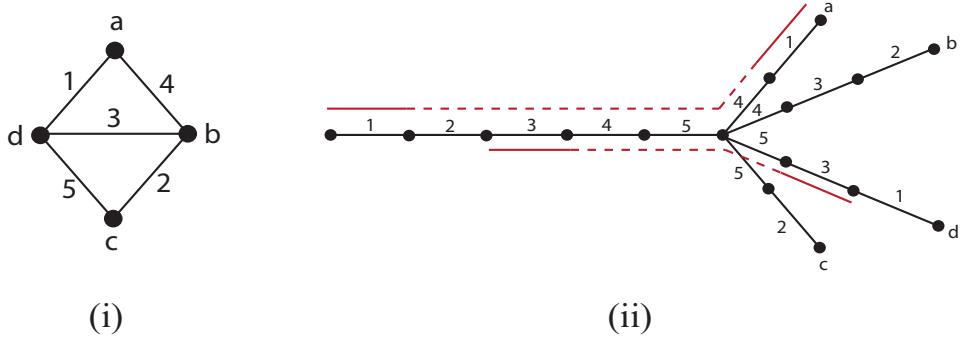}
  \end{center}
  \caption{\label{fig:red} (i) shows an instance of
    the vertex cover problem, and (ii) is the
    corresponding PTC instance. The numbers on the edges are the
    priority demands corresponding to the edge numbers in the
    graph. Figure (ii) also shows two 
    segments: $s^1_a$ and $s^3_d$, having supplies $1$ and $3$
    respectively.  Dashed line means that these segments do not have
    enough supply to cover the edges.}
\end{figure}

The following lemma along with the APX-hardness of the vertex cover
problem in bounded degree graphs, and the fact that in the latter any
vertex cover is of size $\Omega(n)$, leads to the APX-hardness of the
PTC problem.

\begin{lemma}
The optimum PTC of the above instance is $m+k$, where $k$ is the size of the optimum vertex cover of $G$.
\end{lemma}
\begin{proof}
  Firstly note that we may assume that in any optimal PTC, for any
  edge $i=(v,w)$, we will have exactly one of $s^i_v$ or $s^i_w$ in
  the solution. We need to have one since these are the only two
  segments that cover edge $e_i$ in the tree. Instead of picking
  both, we can remove one, say $s^i_w$, from the solution and pick the
  corresponding vertex segment $t_w$ instead, at no increase of
  cost. Therefore, there are exactly $m$ edge segments picked in any
  optimal PTC solution.

  Now note that these $m$ edge segments uniquely correspond to an
  orientation of the edges in $G$; if for edge $i=(v,w)$, $s^i_v$ is
  chosen in the solution, the edge $(v,w)$ is oriented from $w$ to
  $v$.  In this orientation, if there is a {\em sink} (a vertex with
  all edges incident to it) $v$, then note that all the edges in the
  bristle corresponding to $v$ have also been covered. Thus, the
  number of vertex segments required to cover the remaining edges of
  the tree, is precisely the number of {\em non-sinks} in this
  orientation. In particular, the optimal PTC corresponds to the
  orientation that minimizes the number of non-sinks.

  The proof is complete by noting that non-sinks form a vertex cover;
  this is because each edge is oriented away from some non-sink, and
  is thus incident to it. Furthermore, given a vertex cover, there
  exists an orientation with precisely these vertices as
  non-sinks. Orient the edges towards the complement of the vertex
  cover (the independent set) - the complement is precisely the set of
  sinks, and thus the vertex cover is precisely the set of non-sinks.
\end{proof}

\begin{proof}[Proof of Theorem \ref{thm:ptc-hard}]
  Suppose the degrees of $G$ are all $B$, a constant.  Note that the
  vertex cover of this graph is at least $m/B = n/2$.  The
  APX-hardness implies that it is NP-hard to distinguish between the
  case when the vertex cover is $c_1n$ or $c_2n$ where $c_2 > c_1 \ge
  1/2$ are certain constants.

  The above lemma therefore implies it is NP-hard to distinguish
  between the cases when the optimum of a PTC is $m+c_1n = (c_1 +
  B/2)n$ and when the optimum is $m+c_2n = (c_2 + B/2)n$. Since
  $B,c_1,c_2$ are constants, we get the APX-hardness.

  (For the interested reader: the APX-hardness of vertex cover of
  bounded degree graphs by Berman and Karpinski \cite{BK98} gives
  $B=4$, $c_1 = 78/152$ and $c_2 = 79/152$, showing it is NP-hard to
  approximate to a factor better than $1.002 $.)
\end{proof}
}

\subsection{An approximation algorithm for PTC}
\ifconf{
We use the exact algorithm for PLC to get the factor $2$ algorithm for PTC.
The crucial idea is the following.  Given an optimum
solution $S^*\subseteq \ess$, we can partition the edge-set $E$ of $T$
into disjoint sets  $E_1, \ldots, E_p$, and partition two copies of $S^*$ into $S_1,\ldots,S_p$,
such that $E_i$ is a  path in $T$ for each $i$,
and $S_i$ is a priority line cover for the path $E_i$. 
Using this, we describe the $2$-approximation algorithm which proves Theorem~\ref{thm:ptc-apx}. \\

\ni
{\em Proof of Theorem \ref{thm:ptc-apx}:}
For any two vertices $t$ (top) and $b$ (bottom) of the tree $T$, such that $t$ is an ancestor of $b$, let $P_{tb}$ be the unique path from $b$ to $t$. Note that $P_{tb}$, together with the restrictions of the segments in \ess  to $P_{tb}$, defines an instance of PLC. Therefore, for each pair $t$ and $b$, we can compute the optimal solution to the corresponding PLC instance using the exact algorithm; let the cost of this solution be $c'_{tb}$. Create an instance of the 0,1-tree cover problem with $T$ and segments 
$\ess' := \{(t,b): t \mbox{ is an ancestor of } b\}$ with costs $c'_{tb}$.
Solve the 0,1-tree cover instance exactly (recall we are in the rooted version) and for the segments $(t,b)$ in $\ess'$ returned, return the solution of the corresponding PLC instance of cost $c'_{tb}$.

One now uses the decomposition above to obtain a solution to the 0,1-tree cover problem $(T,\ess')$ of cost at most $2$ times the cost of $S^*$. This proves the theorem. The segments in $\ess'$ picked are precisely the segments 
corresponding to paths $E_i$, $i=1,\ldots,p$ and each $S_i$ is a solution to the PLC instance. Since we find the optimum PLC, there is a solution to $(T,\ess')$ with costs $c'$ of cost less than total cost of segments in $S_1\cup \cdots \cup S_p$. But that cost is at most twice the cost of $S^*$ since each segment of $S^*$ is in at most two $S_i$'s.
}
\iffull{
The crucial idea is the following.  Given an optimum solution
$S^*\subseteq \ess$, we can partition the edge-set $E$ of $T$ into
disjoint sets $E_1, \ldots, E_p$, and partition two copies of $S^*$
into $S_1,\ldots,S_p$, such that $E_i$ is a path in $T$ for each $i$,
and $S_i$ is a priority line cover for the path $E_i$. Once again, we assume
without loss of generality that the instance is segment-complete.

In particular, we prove the following lemma. Let $\hat{E}_{S^*,j}$ be
the set of edges $e$ such that $j$ is the segment with the highest
supply, among all segments in $S^*$ that cover $e$.  Note that the
union of all $\hat{E}_{S^*,j}$, over all $j\in S^*$, partitions
$E$. Also note that for each edge $e$, there is a unique segment $j$
such that $e\in \hat{E}_{S^*,j}$. 
If there were two, we could replace one of the segments by a sub-segment and still stay feasible.
We call the segment $j$ {\em
  responsible} for $e$.

\begin{lemma}\label{lem:ptc-apx}
  Given an optimal solution $S^* \subseteq \ess$ to a PTC
  instance with tree $T=(V,E)$, there is a partition
  $$ E_1 \cup \ldots \cup E_p = E, $$
  where each $E_i$ is the edge set of a path in $T$ such that for all $j \in S^*$,
  $\hat{E}_{S^*,j} \cap E_i \neq \varnothing$
  for at most two $i \in \{1, \ldots, p\}$.
\end{lemma}

Using this, we describe the $2$-approximation algorithm which proves Theorem~\ref{thm:ptc-apx}. \\

\begin{proof}[Proof of Theorem \ref{thm:ptc-apx}]
  For any two vertices $t$ (top) and $b$ (bottom) of the tree $T$,
  such that $t$ is an ancestor of $b$, let $P_{tb}$ be the unique path
  from $b$ to $t$. Note that $P_{tb}$, together with the restrictions
  of the segments in \ess to $P_{tb}$, defines an instance of
  PLC. Therefore, for each pair $t$ and $b$, we can compute the
  optimal solution to the corresponding PLC instance; let the cost of
  this solution be $c'_{tb}$. Create an instance of the 0,1-tree cover
  problem with $T$ and segments $\ess' := \{(t,b): t \mbox{ is an
    ancestor of } b\}$ with costs $c'_{tb}$.  Solve the 0,1-tree cover
  instance exactly (recall we are in the rooted version) and for the
  segments $(t,b)$ in $\ess'$ returned, return the solution of the
  corresponding PLC instance of cost $c'_{tb}$.
  We now use Lemma \ref{lem:ptc-apx} to obtain a solution to the
  0,1-tree cover problem $(T,\ess')$ of cost at most $2$ times the
  cost of $S^*$. This will prove the theorem. 

For each $E_i$, let $t_i$ and $b_i$ be the end points of $E_i$ with 
$t_i$ being the ancestor of $b_i$. Since $E_i$'s partition the edges,
the segments $(t_i,b_i):i=1,\ldots,p$ is a feasible 0,1-tree cover for $(T,\ess')$.
Define $S_i := \{j\in S^*: e\in E_i \cap \hat{E}_{S^*,j}\}$ to be
  the set of segments responsible for the edges in $E_i$. By
  definition, $S_i$ is a PLC for $E_i$. Thus, the cost of the segments in 
$S_i$ is at least $c'_{t_ib_i}$. Furthermore, Lemma \ref{lem:ptc-apx}
implies that the total cost of the segments in $S_i$ is at most twice
the cost of segments in $S^*$. Therefore, the cost of the feasible
solution to the cover problem in $(T,\ess')$ is at most twice the cost 
of segments in $S^*$.

\end{proof}

\begin{proof}[Proof of Lemma \ref{lem:ptc-apx}]
  We give an algorithm to compute the decomposition. Let $e$ be 
  any of the edges incident to the root of $T$, and let $j_1 \in S^*$
  be the highest-supply segment covering $e$. We then let $E_1$
  be the edges of the path in $T$ corresponding to $j_1$. Removing
  $E_1$ from $T$ yields sub-trees $T_1, \ldots, T_q$. For each
  tree $T_i$ we repeat the above steps, and let 
  \begin{equation}\label{eq:part}
    E_1, \ldots, E_p 
  \end{equation}
  be the final partition; let $j_i \in S^*$ be the segment 
  corresponding to edge-set $E_i$. Note that for $q < q'$,
  $\hat{E}_{S^*,j_q} \cap E_{q'}$ is empty. This is because $E_{q'}$
is a subset of edges which are not in $j_{q'-1},\ldots,j_1$.

  Consider a segment $j \in S$, and let $1 \leq i \leq p$ be 
  smallest such that $\hat{E}_{S^*,j} \cap E_i \neq \varnothing$, and
  assume that $\hat{E}_{S^*,j} \cap E_q \neq \varnothing$ for some
  $i < q \leq p$; choose $q$ smallest with this property. We
  claim that $j_q=j$, 
 and hence for all
  $q < q' \leq p$ we have $\hat{E}_{S^*,j} \cap E_{q'} = \emptyset$.
Thus, $\hat{E}_{S^*,j}$ has non-empty intersection only with 
$E_i$ and $E_q$.

  Let $e \in E_{S^*,j} \cap E_i$, and let $f \in E_{S^*,j} \cap E_q$ be
  two edges in different parts of the partition such that $j$ is responsible for both. 
  As both $e$ and $f$ are edges on $j$, and since 
  $i < q$, it follows that $f$ is a descendant of $e$ in tree $T$.
  Let $g$ be the topmost edge of $E_q$; clearly, $g$ is on the
  $e,f$-path in $T$. By the decomposition algorithm, segment $j_q$ is
  the highest-supply
  segment covering edge $g$. As $j$ contains $g$, this means 
  that the supply of $j_q$ is at least that of $j$. Finally, since $f$ is 
  on $j_q$, $j_q$ covers $f$ as well. But this means that 
  $j_q=j$ as $j$ is responsible for $f$.
\end{proof}

\subsection{Canonical LP relaxation of PTC: Integrality Gap}

In this section, we prove Theorem \ref{thm:ptc-gap}, by showing that 
the canonical LP relaxation of unweighted PTC is at most $6$. 
Recall the PTC LP.
\begin{align}\label{lp:primal-ptc}
  \min ~~ \left\{\sum_{s\in \ess} c_sx_s: ~~ \forall e\in E: \sum_{s:
      s~covers~ e} x_s \ge 1; ~~ x_s\ge 0, \forall s\in \ess\right\}
\end{align}

\begin{proof}[Proof of Theorem \ref{thm:ptc-gap}]
  The idea of the proof is the following: as in the factor
  $2$-approximation for PTC, we decompose the edge set of the tree
  into disjoint sets $E_1,\cdots,E_p$, such that each $E_i$ induces a
  path.  We will abuse notation and refer to the $E_i$'s as paths.
  Furthermore, we take any feasible solution $x$ of
  \eqref{lp:primal-ptc} and obtain $p$ fractional solutions
  $x^{(1)},\ldots,x^{(p)}$ such that $x^{(i)}$ is 
  a feasible
  fractional solution to \eqref{lp:primal} for the PLC instance on the
  path $E_i$. We will guarantee that
  $$\sum_{i=1}^p\sum_{j\in\ess} x^{(i)}_j\le 3\sum_{j\in \ess} x_j. $$
  The theorem then follows from Theorem \ref{thm:plc-gap}.
  
  \piccaptioninside
  \piccaption{\label{fig:frag} The figure shows a fragment $E_i$, its
   parent $E_r$, and two children $E_s$ and $E_t$. The segments
    $j_1, j_2$, and $j_3$ are local for $E_i$, and segment $j_4$ is
    global. In particular, $j_4$ is an $i,t$-global segment.}

  \parpic(7cm,8cm)[fr]{
   \includegraphics[scale=.85]{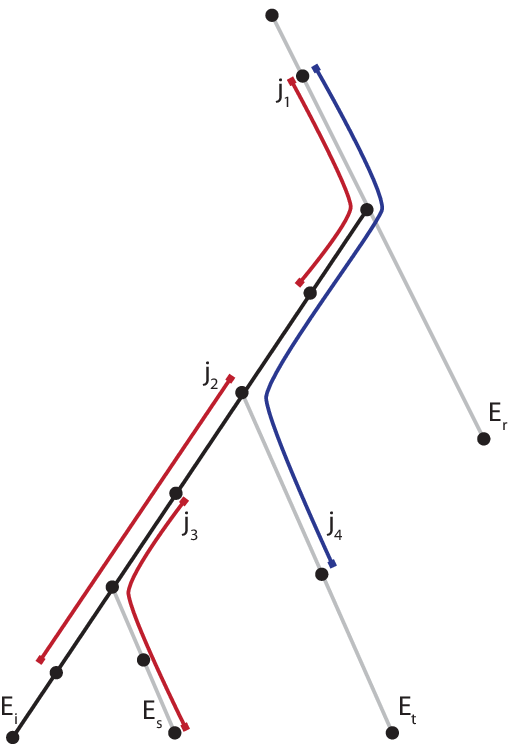}}

  Unlike in the argument used in the previous section where the
  decomposition into paths depended on $S^*$, the decomposition into
  disjoint paths that we use here is universal.  Each path $E_i$ will
  end at a unique leaf, and $p$ in \eqref{eq:part} will now be the
  number of leaves of $T$. Let $E_1$ be {\em any} path from the root
  to a leaf. Delete $E_1$ from the tree to get a series of sub-trees.
  Recursively, obtain $E_2$ to $E_p$. We call a path $E_i$ a {\em
    child} of $E_q$, if the starting point of $E_i$ lies on $E_q$.

  Let $x$ be any feasible fractional solution of \eqref{lp:primal-ptc}
  and let $S^*$ be the support of $x$, that is, $S^* = \{j \,:\, x_j >
  0\}$.  Fix a path $E_i$ and say that a segment $j\in S^*$ {\em
    intersects} $E_i$ if $j$ covers an edge in $E_i$ 
  A segment $j$ that intersects $E_i$ is called {\em local} for $E_i$
  if either the first or the last edge covered by $j$ lies in $E_i$.
  A segment $j$ that intersects $E_i$ is called {\em global} for
  $E_i$, otherwise. Figure \ref{fig:frag} illustrates this.

  Let $j$ be a global segment for $E_i$, and let $e$ be the
  first edge contained in $j$ {\em after} $E_i$.  If $e\in E_q$, we
  call $s$ an {\em $iq$-global} segment. Observe that $E_q$ is a child
  of $E_i$.  Thus an $iq$-global segment {\em enters} $E_i$ and {\em
    exits} via $E_q$.  Note that $iq$-global segments, over all $q$ such that $E_q$ is a child of $E_i$,
  partition all global segments for $E_i$. Also note that an
  $iq$-global segment could also be a $i'q'$-global segment for some
  other $i',q'$.

  Now we are ready to define the fractional solution $x^{(i)}$ that
  will be feasible for \eqref{lp:primal} for the PLC instance on
  $E_i$.  Firstly for all segments $j$ that are local for $E_i$, let
  $x^{(i)}_j = x_j$. Next, we take care of segments that are global
  for $E_i$.  For each child $E_q$ of $E_i$, order all the $iq$-global
  segments in non-increasing order of supply: $\{j_1,\ldots,j_r\}$.
  Let $l$ be such that
  $$x_{j_1} + \cdots + x_{j_l} \le 1 \mbox{ and } x_{j_1} + 
      \cdots + x_{j_l} + x_{j_{l+1}} >  1 $$
  If no such $l$ exists, then $l=r$.  Define $x^{(i)}_{j_k} = x_{j_k}$
  for $1\le k\le l$. If $l< r$, then let $x^{(i)}_{j_{l+1}} = 1 -
  \sum_{k=1}^l x^{(i)}_{j_k}$.

\begin{claim}
  $x^{(i)}$ is feasible for \eqref{lp:primal} for the PLC instance on $E_i$.
\end{claim}
\begin{proof}
  Pick any edge $e\in E_i$. Look at all segments $j\in S^*$ that
  cover $e$. These segments are either local for $e$ or global for
  $e$. If $j$ is local, there is a corresponding segment in the
  support $x^{(i)}$ of the same value. Furthermore for any $q$,
  $$\sum_{j: j \mbox{~is $iq$-global}, j \mbox{ covers } e} x^{(i)}_j 
    \ge \min\{1, \sum_{j: j \mbox{~is $iq$-global}, j \mbox{ covers } e} x_j\} $$
  In any case, $e$ is covered by $x^{(i)}$ at least to the extent it
  is covered by $x$, which implies $x^{(i)}$ is feasible.
\end{proof}

\begin{lemma}
  $\sum_{i=1}^p \sum_{j\in\ess} x^{(i)}_j \le 3\sum_{j\in \ess} x_j $
\end{lemma}

\begin{proof}
  Each segment $j\in S^*$ is local for at most two paths $E_i$ and
  $E_q$. Thus the contribution to the LHS by local segments for some
  path $E_i$ is exactly $2\sum_{j\in \ess} x_j $.

  Furthermore, for every parent-child pair $E_i$ and $E_q$ that
  induces an $iq$-global segment for $E_i$, we increase the LHS by at
  most $1$. The number of such pairs is at most the number of leaves
  in $T$. 
   The proof is complete by noting that $\sum_{j\in \ess}
  x_j$ is at least the number of leaves in $T$.
\end{proof}

To complete the proof of the theorem, note that from Theorem
\ref{thm:plc-gap} we know there exists for each $E_i$, a set of segments
$S_i$ such that $|S_i| \le 2\sum_{j\in \ess} x^{(i)}_j$. The union of
all such $S_i$ forms a valid PTC of cardinality at most $6\sum_{j\in
  \ess} x_j $.
\end{proof}
}
\subsection{Priority Tree Cover and Geometric Covering Problems}
\ifconf{
We sketch  how the PTC problem can be encoded as a rectangle cover problem.
To do so, an auxiliary problem is defined, which we call $2$-PLC. \\

\vspace{-2mm}
\ni
{\bf $2$-Priority Line Cover (2-PLC)} The input is similar to PLC, except each segment and edge has now 
an ordered pair of priorities, and a segment covers an edge it contains iff each of the priorities of the segment
exceeds the corresponding priority of the edge. The goal, as in PLC, is to find a minimum cost cover.

It is not too hard to show 2-PLC is a special case of rectangle cover. The edges correspond to points
in 3 dimension and segments correspond to rectangles in $3$-dimension;
dimensions encoded by the linear coordinates on the line, and the two priority values.
In general, $p$-PLC can be shown to be a special case of $(p+1)$-dimensional rectangle cover. 

What is more involved is to show PTC is a special case of 2-PLC. To do so, we run two DFS orderings on the 
tree, where the order in which children of a node are visited is completely opposite in the two DFS orderings.
The first ordering gives the order in which these edges must be placed on a line. The second gives one of the 
priorities for the edges. The second priority of the edges comes from the original priority in PTC. It can be shown
that the segments priorities can be so set that the feasible solutions are precisely the same in both the instances
proving Theorem \ref{thm:ptc-geom}.
}
\iffull{
In this section, we show that the PTC problem is a special case of covering a set 
of points in $3$-dimension by axis-parallel rectangles (cuboids). In particular we prove 
Theorem \ref{thm:ptc-geom}. We go in two steps. We first define a problem, that we 
call $2$-Priority Line Cover and show that the PTC problem is a special case of $2$-PLC. Subsequently,
we show $2$-PLC is a special case of $3$-dimensional rectangle cover.
We start with a definition of $2$-PLC. \\

\noindent
{\bf $2$-Priority Line Cover (2-PLC).} The input is  a line $T=(V,E)$, and a collection of segments 
$\ess \subseteq V\times V$ with costs $c_j$ for each $j\in \ess$. Furthermore, each segment $j$
has a priority supply vector in {\em two} dimensions, denoted as $(s^1_j,s^2_j)$, and each edge $e$ 
has a priority demand vector in two dimensions, denoted as $(\pi^1_e,\pi^2_e)$. A segment $j$ 
covers $e$ iff $j$ contains $e$ and $s^i_j \ge \pi^i_e$ for both $i=1,2$. The goal is to find the minimum
cost collection of segments that cover every edge. \\

\noindent
It is easy to see that PLC is a special case of $2$-PLC. Somewhat surprisingly, PTC is a special case 
of $2$-PLC as well.

\begin{lemma}
Any instance of PTC can be encoded as an instance of 2-PLC with the same solution set.
\end{lemma}
\begin{proof}
Given a rooted tree $T=(V,E)$, we perform two different depth first traversals
to get two different orderings on the edges $E$. One such ordering will define the line
of the 2-PLC instance, the other will define the first coordinates of the priority demand vectors 
of the edges.

In a depth first traversal of a tree, at every step we move from a vertex to one of its children, if any.
Our two different traversals will be defined by two different choices of moving to a child-vertex.
For every vertex $v$ of the tree, consider a total order $\sigma_v$ on its children. One such order that is convenient
to keep in mind is the following; given a drawing of the tree, the total order of the children is from left to right.
Let $\sigma^R_v$ be the {\em opposite} total order. The two depth first traversals are obtained by running 
with $\sigma_v$'s and $\sigma^R_v$'s, respectively. Figure \ref{fig:illus-dfs} illustrates the two orders 
with the ordering $\sigma_v$ at every vertex $v$ being from left-to-right, and $\sigma^R_v$ being from
right-to-left.

\begin{figure}[h]
    \begin{center}
      \includegraphics[scale=0.55]{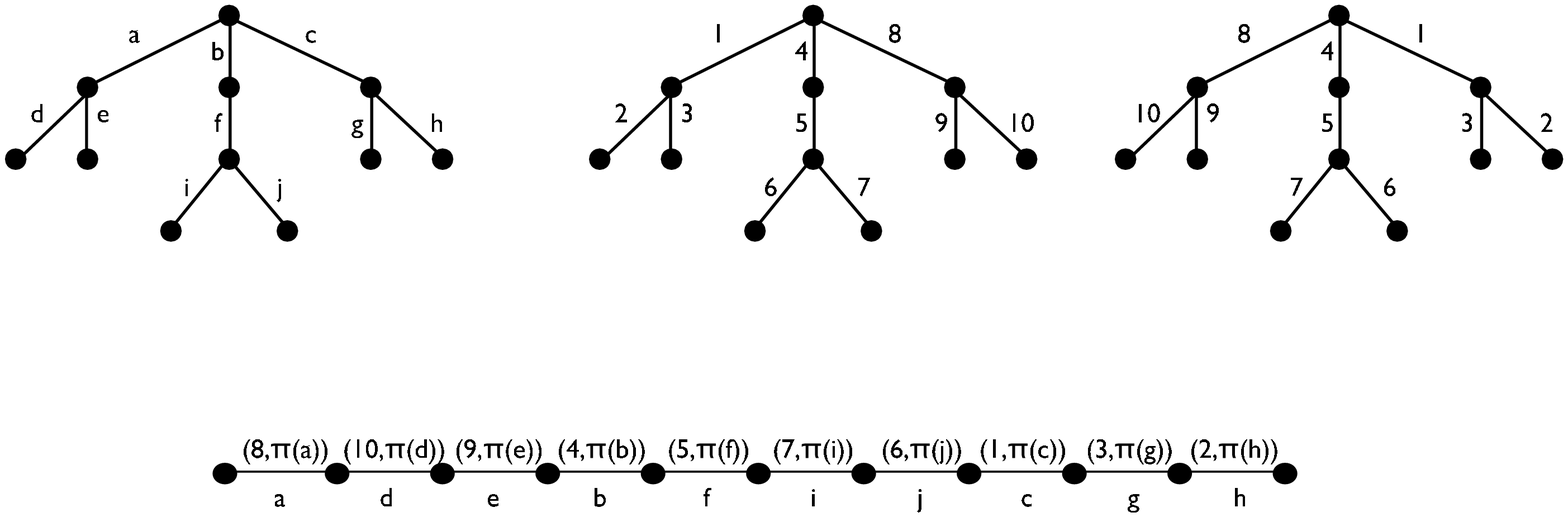}
      \caption{The left most tree is the original tree, the second and third are the two depth first traversals. The line below
      shows the line in the 2-PLC instance.}
      \label{fig:illus-dfs}
    \end{center}
\end{figure}

  Let the two traversals return orderings $\mu$ and $\mu^R$ on the
  edges of the tree. The crucial observation is the following: for any
  vertex $v$, let $(v_1,\ldots,v_k)$ be the children in the $\sigma_v$
  order; then $\mu(v,v_1) < \mu(v,v_2) < \cdots < \mu(v,v_k)$, and
  thus, $\mu^R(v,v_1) > \cdots > \mu^R(v,v_k)$.

  Now we are ready to describe the 2-PLC instance. The line is defined
  by the edges of the tree ordered w.r.t. $\mu$.  That is, the order
  of the edges is $(e_1,\ldots,e_m)$ such that $\mu(e_1) < \mu(e_2) <
  \cdots < \mu(e_m)$.  The priority demand vector of an edge $e$ of
  the tree is $(\mu^R(e),\pi_e)$. Consider a segment $j = (u,v)$ such
  that $u$ is a descendant of $v$ in the PTC instance. We identify two
  specific tree edges contained in $j$: the parent-edge
  $(u,u')$ of $u$, and the edge $(v,v')$ between node $v$ and its 
  unique child $v'$ that is on the $u,v$-path in $T$.
  By the depth-first property, we get $\mu(v,v') \leq \mu(u,u')$.
  The corresponding segment in the 2-PLC instance, also denoted as
  $j$, contains all the edges from $\mu(v,v')$ to $\mu(u,u')$. The
  priority supply vector of $j$ is $(\mu^R(u,u'),s_j)$.

\begin{claim}
  For any segment $j$, the set of edges covered by $j$ in the 2-PLC
  instance is precisely the set of edges covered in the PTC instance.
\end{claim}
\begin{proof}
  Let $e$ be an edge covered by $j$ in the PTC instance. Since $e$ is
  contained in the path from $u$ to $v$ in the tree, by property of
  depth first traversals we get, $\mu(v,v') \leq \mu(e) \leq
  \mu(u,u')$ and $\mu^R(e) \leq \mu^R(u,u')$.  The first pair of
  inequalities implies $e$ lies in the segment $j$ in the 2-PLC
  instance, the second implies that $\pi^1_e \leq s^1_j$. Since $e$ is
  covered by $j$ in the PTC, we also get $\pi^2_e = \pi_e \le s_j =
  s^2_j$. Thus, $e$ is covered by $j$ in the 2-PLC instance.

  Let $e$ be an edge covered by $j$ in the 2-PLC instance. Since $e$
  lies in $j$, we conclude $\mu(v,v') \leq\mu(e) \leq \mu(u,u')$. This
  implies either (a) $e$ lies on the path from $u$ to $v$ in the tree,
  or, (b) there is a node $w$ on the $u,v$-path in the tree, and a
  child $z$ of $w$ that is not on this path such that $e$ is contained
  in the subtree defined by edge $(w,z)$. 
  
  Note, that in case (b) the depth-first traversal for order $\sigma$
  visits edge $(z,w)$ {\em before} edge $(u,u')$. This implies that 
  the second dfs traversal for order $\sigma^R$
  visits $(z,w)$ {\em after} $(u,u')$. Since $(z,w)$ is visited before
  $e$ in both traversals, we must therefore have
  $\mu^R(e) > \mu^R(u,u')$, and this implies $s^1_j <
  \pi^1(e)$ which is impossible since $j$ covers $e$. Thus, case (b)
  is not possible, and $e$ lies on the path fro $u$ to $v$ on the
  tree.  Furthermore, we have $s_j = s^2_j \ge \pi^2_e = \pi_e$, and so
  $j$ covers $e$ in the PTC instance as well. 
\end{proof}  
\end{proof}

Now we show that $2$-PLC is a special case of $3$-dimensional
rectangle cover. This is not to hard to see.  We assume the edges of
the line are numbered $(1,2,\ldots,m)$. For edge $e$ numbered $e_i$,
we associate a point in $3$ dimensions with coordinates
$(i,\pi^1_e,\pi^2_e)$. For each segment $j = (a,b)$, we have a
rectangle associated.  In fact, these rectangles have are unbounded in
the negative $y$ and $z$ coordinates. The other $4$ bounding
half-spaces are $x \ge a$, $x \le b$, $y\le s^1(j)$ and $z \le
s^2(j)$. It is not too hard to see a rectangle corresponding to a
segment $j$ contains a point corresponding to an edge $e$ iff $j$
covers $e$ in the 2-PLC instance. This completes the proof of Theorem \ref{thm:ptc-geom}.
}
\section{Concluding Remarks}

In this paper we studied column restricted covering integer
programs. In particular, we studied the relationship between CCIPs and
the underlying 0,1-CIPs. We conjecture that the approximability of a
CCIP should be asymptotically within a constant factor of the
integrality gap of the original 0,1-CIP. We couldn't show this;
however, if the integrality gap of a PCIP is shown to be within a
constant of the integrality gap of the 0,1-CIP, then we will be
done. At this point, we don't even know how to prove that PCIPs of
special 0,1-CIPS, those whose constraint matrices are totally
unimodular, have constant integrality gap. Resolving the case of PTC
is an important step in this direction, and hopefully in resolving our
conjecture regarding CCIPs.

\bibliography{cover}
\bibliographystyle{plain}

\end{document}